\documentclass{article}
\usepackage[a4paper,margin=2cm,footskip=.5cm]{geometry}
\usepackage{setspace}
\usepackage{graphicx, epstopdf, parskip, latexsym, amsmath, amssymb, mathtools, times, amsthm, url, multicol, rotating, natbib, algorithm, algpseudocode, bm}
\usepackage[english]{babel}
\usepackage[utf8]{inputenc}
\usepackage{latexsym}
\usepackage{fancyhdr}
\usepackage{color}
\usepackage{hyperref}

\theoremstyle{plain}
\newtheorem{thm}{Theorem}

\newtheorem{Cor}{Corollary}

\theoremstyle{definition}
\newtheorem{Def}{Definition}
\newtheorem{Assum}{Assumption}

\theoremstyle{remark}
\newtheorem{Rem}{Remark}
\newtheorem{Eg}{Example}

\DeclareMathOperator{\Var}{Var}
\DeclareMathOperator{\Cov}{Cov}
\DeclareMathOperator{\Corr}{Corr}

\DeclareMathOperator{\Leb}{Leb}
\DeclareMathOperator{\argmin}{argmin}
\allowdisplaybreaks

\title{Bridging between short-range and long-range dependence \\with mixed spatio-temporal Ornstein-Uhlenbeck processes}
\author{MICHELE NGUYEN AND ALMUT E. D. VERAART\\
	\textit{Department of Mathematics, Imperial College London}
	}
\providecommand{\keywords}[1]{\textbf{\textit{Keywords:}} #1}
\providecommand{\subjclass}[1]{\textbf{\textit{Mathematics Subject Classification:}} #1}

\date{}

\begin{document}

\maketitle
\pagenumbering{arabic}

\thispagestyle{fancy}

\begin{abstract}
While short-range dependence is widely assumed in the literature for its simplicity, long-range dependence is a feature that has been observed in data from finance, hydrology, geophysics and economics. In this paper, we extend a L\'evy-driven spatio-temporal Ornstein-Uhlenbeck process by randomly varying its rate parameter to model both short-range and long-range dependence. This particular set-up allows for non-separable spatio-temporal correlations which are desirable for real applications, as well as flexible spatial covariances which arise from the shapes of influence regions. Theoretical properties such as spatio-temporal stationarity and second-order moments are established. An isotropic $g$-class is also used to illustrate how the memory of the process is related to the probability distribution of the rate parameter. We develop a simulation algorithm for the compound Poisson case which can be used to approximate other L\'evy bases. The generalised method of moments is used for inference and simulation experiments are conducted with a view towards asymptotic properties.
\end{abstract}

\keywords{Long range dependence, Ornstein-Uhlenbeck process, spatio-temporal, compound Poisson, generalised method of moments.} \\
\subjclass{60G10, 60G55, 60G60, 62F10, 62F12, 62M30}

\section{Introduction}

L\'evy-driven Ornstein-Uhlenbeck (OU) processes are popular modelling tools in finance due to their mean-reverting properties and ability to exhibit non-Gaussianity. To encompass the long memory that has been observed in time series of financial volatility, an extension towards randomly-varying rate parameters was introduced in \cite{Barndorff2001}. This results in a superposition of OU processes, i.e.~a supOU process, which can be seen as another materialisation of the idea in data traffic modelling and hydrology that long memory arises from a hierarchy or aggregation of processes \cite[]{DOT2002}. The supOU process itself has been studied extensively in its univariate and multivariate contexts as well as in its extremal properties \cite[]{BS2011, FK2007}.   
\\
OU processes have also been used in the spatio-temporal setting. While \cite{TLB2004} studied a product of a one-dimensional spatial OU process with a temporal OU process, \cite{BD2001} considered a multivariate OU process whose vector components correspond to different spatial locations. In the latter, the authors were motivated by environmental epidemiology and used the OU process as the stochastic intensity process of a log-Gaussian Cox process. To model turbulence, \cite{BS2003} defined a class of spatio-temporal OU (STOU) processes as stochastic integrals with L\'evy noise. This can be seen as a direct spatio-temporal extension of the L\'evy-driven OU processes used in finance. We call a random field $\{Y_{t}(\mathbf{x})\}$ in space-time $\mathcal{X}\times\mathcal{T} = \mathbb{R}^{d}\times \mathbb{R}$ for some $d\in\mathbb{N}$ a \textit{STOU process} if:

\begin{equation*}
Y_{t}(\mathbf{x}) = \int_{A_{t}(\mathbf{x})} \exp(-\lambda(t-s)) L(\mathrm{d}\boldsymbol{\xi}, \mathrm{d}s),
\end{equation*}

where $\lambda>0$ and $L$ is a homogeneous L\'evy basis with finite second moments. The integration set or ambit set, $A_{t}(\mathbf{s})\subset  \mathcal{X} \times\mathcal{T}$, can be interpreted as a causality cone in physics and satisfies the following conditions: 

\begin{equation}
\begin{cases}
A_{t}(\mathbf{x}) = A_{0}(\mathbf{0}) + (\mathbf{x}, t), &\text{ (Translation invariant)}
\\
A_{s}(\mathbf{x}) \subset A_{t}(\mathbf{x}), \forall s < t,
\\
A_{t}(\mathbf{x}) \cap (\mathcal{X} \times (t, \infty)) = \emptyset. &\text{ (Non-anticipative)}
\end{cases}
\label{eqn:ambitassumptions}
\end{equation}

Further studies have shown that this class of processes exhibits exponential temporal correlation just like the temporal OU process and boosts flexible spatial correlation structures which are determined by the shape of the ambit set  \cite[]{NV2016}. In addition, non-separable covariances, which are desirable in practice, can be obtained. 
\\
In this paper, we extend the STOU processes by mixing the rate parameter $\lambda$. This will enable us to bridge between short-range and long-range dependence structures in space-time. A \textit{mixed spatio-temporal OU (MSTOU) process} is defined by:
\begin{equation}
Y_{t}(\mathbf{x}) = \int_{0}^{\infty}\int_{A_{t}(\mathbf{x})} \exp(-\lambda(t-s)) L(\mathrm{d}\boldsymbol{\xi}, \mathrm{d}s, \mathrm{d}\lambda).
\label{eqn:MSTOU}
\end{equation}
Now, $L$ is a L\'evy basis over the product space of space-time and the $\lambda$ parameter space. In addition, it is no longer homogeneous since we typically associate the parameter space with a probability distribution. Depending on the parameters of this distribution, the process has either short-range or long-range dependence. This extension of STOU processes will be useful for applications where long memory has been observed, for example, in hydrology, geophysics and economics \cite[]{FMAA2008, DOT2002}. 

\paragraph{Outline}
In the next section, we introduce the background required to understand the construction of (\ref{eqn:MSTOU}). In Section \ref{sec:Prop}, we derive the key theoretical properties of the MSTOU process. This includes spatio-temporal stationarity and second-order moments. Particular focus is given to the isotropic $g$-class and we show that long memory can be obtained for specific parameter ranges of the distribution of $\lambda$. By way of an example, we contrast the MSTOU process to another way of defining superpositions of STOU processes which is related to the well-known continuous autoregressive (CAR) process. Unlike the MSTOU process, this alternative definition does not model temporal long memory. In Section \ref{sec:Sim}, we look at the case where $L$ is compound Poisson and simulate from the MSTOU process. Unlike the discrete convolution algorithms for the STOU processes in \cite{NV2016}, we no longer have kernel discretisation error and only have ambit set approximation error that stems from the kernel truncation. The simulation method can also be used to give second-order approximations for other L\'evy bases. In Section \ref{sec:Infer}, we apply the generalised method of moments (GMM) to an MSTOU process. Simulation experiments are conducted to illustrate the finite sample behaviour as well as to provide a view towards to the asymptotic properties of these estimators. Finally, we conclude and discuss further directions for research in Section \ref{sec:Conclusion}. 

\section{Preliminaries} \label{sec:Prelim}

To understand the definition of an MSTOU process in (\ref{eqn:MSTOU}), we rely on the $\mathcal{L}_{0}$ integration theory in \cite{RR1989}. Let $S =  \mathbb{R}^{d}\times\mathbb{R} \times (0, \infty)$, the product space of space-time and the $\lambda$ parameter space. Further denote the Borel $\sigma$-algebra of $S$ by $\mathcal{S} = \mathcal{B}(\mathcal{S})$ and let $\mathcal{B}_{b}(S)$ contain all its Lebesgue-bounded sets. Then, a L\'evy basis is defined as follows \cite[]{BBV2012}: 
\\
\begin{Def}[L\'evy basis] \hfill \\
$L$ is a \textit{L\'evy basis} on $(S, \mathcal{S})$ if it is an independently scattered and infinitely divisible random measure. This means that:
\begin{enumerate}
\item $L = \{L(E): E \in \mathcal{B}_{b}(S)\}$ is a set of $\mathbb{R}$-valued random variables such that for a sequence of disjoint elements of $\mathcal{B}_{b}(S)$, $\{E_{i}: i \in \mathbb{N}\}$:
\begin{itemize}
\item $L(\bigcup_{j = 1}^{\infty} E_{j}) = \sum_{j = 1}^{\infty} L(E_j)$  almost surely when $\bigcup_{j = 1}^{\infty} E_{j} \in \mathcal{B}_{b}(S)$;
\item and for $i\neq j$, $L(E_{i})$ and $L(E_{j})$ are independent. 
\end{itemize}
\item Let $B_{1}, ..., B_{m} \in \mathcal{B}_{b}(S)$ for finite $m\in\mathbb{N}$. The random vector $\mathbf{L} = (L(B_{1}), ..., L(B_{m}))$ is infinitely divisible, i.e.~for any $n \in \mathbb{N}$, there exists a law $\mu_{n}$ such that the law of $\mathbf{L}$ can be expressed as $\mu = \mu_{n}^{*n}$, the n-fold convolution of $\mu_{n}$ with itself.
\end{enumerate}
\end{Def}
Since $L(E)$ corresponds to an infinitely divisible random variable for $E\in\mathcal{B}_{b}(S)$, it obeys a L\'evy-Khintchine (L-K) formula and its cumulant generating function can be written as:
\begin{equation}
C\{\theta \ddagger L(E)\} = \log\mathbb{E}\left[\exp\left(i\theta L(E)\right)\right] = i\theta a^{*}(E) - \frac{1}{2}\theta^{2}b^{*}(E) + \int_{\mathbb{R}}\left(e^{i\theta z } - 1 - i\theta z \mathbf{1}_{|z|\leq 1}\right) n(\mathrm{d}z, E), \label{eqn:LBasisLK}
\end{equation}
where $a^{*}$ is a signed measure on $\mathcal{B}_{b}(S)$, $b^{*}$ is a measure on $\mathcal{B}_{b}(S)$, and $\mathbf{1}_{|z|\leq1} = 1$ for $|z|\leq 1$ and $0$ otherwise. The generalised L\'evy measure $n(\cdot, \cdot)$ is such that for fixed $\mathrm{d}z$, $n(\mathrm{d}z, A)$ is a measure on $\mathrm{B}_{b}(S)$, while for fixed $A\in\mathcal{B}_{b}(S)$, $n(\mathrm{d}z, A)$ is a L\'evy measure, i.e.~it satisfies $\int_{\mathbb{R}}\min(1, z^2)n(\mathrm{d}z, A) < \infty$. Note that the logarithm used in (\ref{eqn:LBasisLK}) is the distinguished logarithm (see page 33 of \cite{Sato1999}). 
\\
In \cite{RR1989}, the authors relate a L\'evy basis to its control measure: 
\\
\begin{Def}[Control measure] \hfill \\
Let $L$ be a L\'evy basis satisfying the L-K formula in (\ref{eqn:LBasisLK}). We define the measure $\tilde{c}$ by:
\begin{equation*}
\tilde{c}(E) = |a^{*}|(E) + b^{*}(E) + \int_{\mathbb{R}}\min(1, z^2)n(\mathrm{d}z, E), 
\end{equation*}
where $E\in\mathcal{B}_{b}(S)$ and $|\cdot|$ denotes total variation. By further requiring $\tilde{c}$ to be $\sigma$-finite, we obtain the \textit{control measure} of $L$.  
\end{Def}
Similar to how an infinitely divisible random variable is characterised by its L-K characteristic triplet, this control measure helps us define the characteristic quadruplet of the L\'evy basis $L$:
\\
\begin{Def}[Characteristic quadruplet and the L\'evy seed]  \label{def:Lseed} \hfill \\
Let $L$ be a L\'evy basis satisfying the L-K formula in (\ref{eqn:LBasisLK}) and let $\tilde{c}$ be its control measure. Suppose that:
\begin{itemize}
\item the Radon-Nikodym derivatives $a(\mathbf{z}) = \frac{\mathrm{d} a^{*}}{\mathrm{d}\tilde{c}}(\mathbf{z})$ and $b(\mathbf{z}) = \frac{\mathrm{d} b^{*}}{\mathrm{d}\tilde{c}}(\mathbf{z})$ are functions on $S$ and $b$ is non-negative;
\item the Radon-Nikodym derivative $\nu(\mathrm{d}z, \mathbf{z}) = \frac{n(\mathrm{d}z, \cdot )}{\mathrm{d}\tilde{c}}(\mathbf{z})$ is a L\'evy measure on $\mathbb{R}$ for fixed $\mathbf{z}$ as well as a measurable function on $S$ for fixed $\mathrm{d}z$. 
\end{itemize} 
Then for $E\in\mathcal{B}_{b}(S)$, $\int_{E}a(\mathbf{z})\tilde{c}(\mathrm{d}\mathbf{z}) = a^{*}(E)$, $\int_{E}b(\mathbf{z})\tilde{c}(\mathrm{d}\mathbf{z}) = b^{*}(E)$ and $\int_{E}\nu(\mathrm{d}z, \mathbf{z})\tilde{c}(\mathrm{d}\mathbf{z}) = n(\mathrm{d}z, E)$. The \textit{characteristic quadruplet (CQ)} of $L$ is given by $(a, b, \nu(\mathrm{d}z, \cdot), \tilde{c}) = (a(\mathbf{z}), b(\mathbf{z}), \nu(\mathrm{d}z, \mathbf{z}), \tilde{c}(\mathbf{z}))_{\mathbf{z}\in S}$. 
\\
The \textit{L\'evy seed} of $L$ is defined as the random variable $L'(\mathbf{z})$ with the L-K representation: 
\begin{equation*}
C\{\theta \ddagger L'(\mathbf{z})\} = i\theta a(\mathbf{z}) - \frac{1}{2}\theta^{2}b(\mathbf{z}) + \int_{\mathbb{R}}(e^{i\theta z } - 1 - i\theta z \mathbf{1}_{[-1, 1]}(z)) \nu(\mathrm{d}z, \mathbf{z}). 
\end{equation*}
\end{Def}
For MSTOU processes, we set the CQ of $L$ such that $a$ and $b$ are constants, and $\nu(\mathrm{d}z, \cdot) = \nu(\mathrm{d}z)$. This means that the L\'evy seed $L'$ does not depend on $\mathbf{z}\in S$. To allow for different values of $\lambda$ in a possibly continuous way over space-time, we set our control measure to be $\tilde{c}(\mathrm{d}\mathbf{z}) = \mathrm{d}\bm{\xi}\mathrm{d}s \pi(\mathrm{d}\lambda)$ where $\bm{\xi}\in\mathbb{R}^{d}$, $s\in\mathbb{R}$, $\lambda\in(0, \infty)$ and $\int_{0}^{\infty}\pi(\mathrm{d}\lambda) = 1$. This means that the L\'evy basis is homogeneous, i.e.~has stationary distributions, over space-time, but is inhomogeneous over the $\lambda$ parameter space in a manner determined by $\pi$. The latter can be interpreted as the probability measure of the parameter $\lambda$ and we typically assume that it has a density, $f(\lambda)$. This means that we work under the following assumptions:
\vspace{2mm}
\begin{Assum}\label{assum:LCQ}
The CQ of the L\'evy basis in (\ref{eqn:MSTOU}) is $(a, b, \nu(\mathrm{d}z), \mathrm{d}\bm{\xi}\mathrm{d}s f(\lambda)\mathrm{d}\lambda)$ where $a\in\mathbb{R}$, $b\geq 0$, $\nu$ is a L\'evy measure and $f(\lambda)$ is a probability density.
\end{Assum}

Now that we have defined the required L\'evy basis, we summarise how stochastic integrals such as (\ref{eqn:MSTOU}) are constructed. Consider the probability space $(\Omega, \mathcal{F}, P)$. We start with the integral of a simple function in $S$ before taking the limit to measurable functions:
 \\
\begin{Def}[Stochastic integral of a simple function] \hfill \\
Let $\{E_{j}: j = 1,..., m\}$ be a collection of disjoint sets of $\mathcal{B}_{b}(S)$ and let $y_{j} \in \mathbb{R}$ for $j = 1, ..., m$. A simple function on $S$ is given by $g(\mathbf{x}, t, \lambda) = \sum_{j = 1}^{m} y_{j}\mathbf{1}_{E_{j}}(\mathbf{x}, t, \lambda)$. The \textit{stochastic integral of $g$ over $A \in \mathcal{S}$} is defined by $\int_{A}g(\bm{\xi}, s, \lambda) L(\mathrm{d}\bm{\xi},\mathrm{d}s, \mathrm{d}\lambda) = \sum_{j = 1}^{m} y_{j} L(A \cap E_{j})$.
\end{Def}
\vspace{2mm}
\begin{Def}[$L$-integrable functions and their stochastic integrals] \hfill \\
A measurable function $g:(S, \mathcal{S}) \rightarrow (\mathbb{R}, \mathcal{B}(\mathbb{R}))$ is said to be \textit{$L$-integrable} if there exists a sequence $\{g_{n}\}$ of simple functions such that:
\begin{itemize}
\item[(i)] $g_{n} \rightarrow g$ as $n \rightarrow \infty$, almost everywhere with respect to $\tilde{c}$.
\item[(ii)] for every $A \in \mathcal{S}$, the sequence $\{\int_{A}g_{n}(\bm{\xi}, s, \lambda)L(\mathrm{d}\bm{\xi}, \mathrm{d}s, \mathrm{d}\lambda)\}$ converges in probability. 
\end{itemize}
If $g$ is $L$-integrable, we write:
\begin{equation*}
\int_{A}g(\bm{\xi}, s, \lambda) L(\mathrm{d}\bm{\xi}, \mathrm{d}s, \mathrm{d}\lambda) = P-\lim_{n\rightarrow\infty} \int_{A} g_{n}(\bm{\xi}, s, \lambda) L(\mathrm{d}\bm{\xi}, \mathrm{d}s, \mathrm{d}\lambda).
\end{equation*}
The construction is well-defined as the limit does not depend on $\{g_{n}\}$. Here, ``$P-\lim$'' refers to the limit achieved through a convergence in probability.
\end{Def}
\vspace{2mm}
Later, we will use the following result from Theorem 2.7 of \cite{RR1989} to show that the MSTOU processes we construct are well-defined:
\\
\begin{thm}
\label{thm:intconditions}
Let $L$ be a L\'evy basis on $(S, \mathcal{S})$ whose CQ is $(a, b, \nu(\mathrm{d}z), \mathrm{d}\bm{\xi}\mathrm{d}s f(\lambda)\mathrm{d}\lambda)$. The measurable function $g:(S, \mathcal{S}) \rightarrow (\mathbb{R}, \mathcal{B}(\mathbb{R}))$ is $L$-integrable if and only if:
\begin{equation*}
\int_{S} |U(g(\bm{\xi}, s, \lambda))|f(\lambda)\mathrm{d}\bm{\xi}\mathrm{d}s\mathrm{d}\lambda < \infty, \int_{S} b|g(\bm{\xi}, s, \lambda)|^{2}f(\lambda) \mathrm{d}\bm{\xi}\mathrm{d}s \mathrm{d}\lambda< \infty, \text{ and } \int_{S} V_{0}(g(\bm{\xi}, s, \lambda)) f(\lambda)\mathrm{d}\bm{\xi}\mathrm{d}s \mathrm{d}\lambda<\infty,
\end{equation*}
where $U(u) = u a + \int_{\mathbb{R}} \big(\rho(zu) - u\rho(z)\big) \nu(\mathrm{d}z)$, $\rho(z) = z\mathbf{1}_{|z|\leq 1}$, and $V_{0}(u) = \int_{\mathbb{R}} \min(1, |zu|^{2})\nu(\mathrm{d}z)$.
\end{thm}
\vspace{2mm}
\begin{Rem}
Here, we are used a slightly different truncation function from the one used in \cite{RR1989}, i.e.~$\rho(z) = z$ if $|z|\leq 1$ and $\frac{z}{|z|}$ if $|z|>1$. 
\end{Rem}
\vspace{2mm}
Since we are interested in the second order moments of MSTOU processes, we need to make the following assumption:
\vspace{2mm}
\begin{Assum}\label{assum:finite2m}
The L\'evy basis in (\ref{eqn:MSTOU}) has finite second moments.
\end{Assum}
\vspace{2mm}
With this assumption, the integrability conditions simplify:
\vspace{2mm}
\begin{Cor}
\label{cor:intcon}
Under Assumptions \ref{assum:LCQ} and \ref{assum:finite2m}, the MSTOU process is well-defined if:
\begin{equation}
\int_{0}^{\infty} \int_{A_{t}(\mathbf{x})} \exp(-\lambda(t-s)) f(\lambda) \mathrm{d}\bm{\xi} \mathrm{d}s \mathrm{d}\lambda < \infty \text{ and } \int_{0}^{\infty} \int_{A_{t}(\mathbf{x})} \exp(-2\lambda(t-s)) f(\lambda) \mathrm{d}\bm{\xi} \mathrm{d}s \mathrm{d}\lambda < \infty. \label{eqn:simicon}
\end{equation}
\end{Cor}

\section{Properties} \label{sec:Prop}

In the following, assume that Corollary \ref{cor:intcon} holds. In this section, we investigate the theoretical properties of MSTOU processes. The proofs of the results can be found in the Appendix. 

\subsection{Finite-dimensional distribution and stationarity}

The distribution of an MSTOU process is determined by its ambit set $A_{t}(\mathbf{x})$ and the CQ of its L\'evy basis. A summary of this can be obtained through its generalised cumulant functional \cite[]{BBV2012, NV2016}:
\\
\begin{Def}[Generalised cumulant functional]\label{defn:cgf} \hfill \\
For a random field in space-time, $Y = \{Y_{t}(\mathbf{x})\}_{\mathbf{x}\in\mathbb{R}^{d}, t\in\mathbb{R}}$, let $v$ denote any non-random measure for which: 
\begin{equation*}
v(Y) = \int_{\mathbb{R}^{d}\times\mathbb{R}} Y_{t}(\mathbf{x}) v(\mathrm{d}\mathbf{x}, \mathrm{d}t),
\end{equation*}
exists almost surely. The \textit{generalised cumulant functional} (GCF) of $Y$ with respect to $v$ is defined as: $C\{\theta \ddagger v(Y) \} = \log\mathbb{E}\left[\exp\left(i\theta v\left(Y\right)\right)\right]$.  
\end{Def}
\vspace{2mm}
\begin{thm}
\label{thm:GCF}
Let $Y$ be an MSTOU process defined by (\ref{eqn:MSTOU}) and $A = A_{0}(\mathbf{0})$. Suppose that for all $\bm{\xi} \in \mathbb{R}^{d}$, $s \in \mathbb{R}$ and $\lambda \in (0, \infty)$,  
\begin{align*}
h_{A}(\bm{\xi},s, \lambda) &= \int_{\mathbb{R}^{d}\times\mathbb{R}} \mathbf{1}_{A}(\bm{\xi} - \mathbf{x}, s-t) \exp(-\lambda(t-s)) v(\mathrm{d}\mathbf{x}, \mathrm{d}t) < \infty,
\end{align*}
and $h_{A}(\bm{\xi}, s, \lambda)$ is integrable with respect to the L\'evy basis $L$. Then, the GCF of $Y$ with respect to $v$ can be written as:
\begin{equation}\label{eqn:MSTOULevy}
\begin{split}
C\{\theta \ddagger v(Y)\} &=  i\theta a  \int_{S}  h_{A}(\bm{\xi},s, \lambda)f(\lambda)\mathrm{d}\bm{\xi}\mathrm{d}s\mathrm{d}\lambda - \frac{1}{2} \theta^{2}b  \int_{S}  h^{2}_{A}(\bm{\xi},s, \lambda)f(\lambda)\mathrm{d}\bm{\xi}\mathrm{d}s\mathrm{d}\lambda \\
&+   \int_{S} \int_{\mathbb{R}} \left(\exp(i\theta h_{A}(\bm{\xi},s, \lambda)z) - 1 - i\theta h_{A}(\bm{\xi},s, \lambda)z\mathbf{1}_{|z|\leq 1}\right) \nu(\mathrm{d}z)  f(\lambda)\mathrm{d}\bm{\xi}\mathrm{d}s\mathrm{d}\lambda,
\end{split}
\end{equation}
where $(a, b, \nu(\mathrm{d}z), \mathrm{d}\bm{\xi}\mathrm{d}s f(\lambda)\mathrm{d}\lambda)$ is the CQ of $L$. 
\end{thm}
\vspace{2mm}
For the marginal and joint distributions of MSTOU processes, we  use $v(\mathrm{d}\mathbf{x}, \mathrm{d}t) = \theta_{1}\delta_{t_{1}}(\mathrm{d}t)\delta_{\mathbf{x}_{1}}(\mathrm{d}\mathbf{x}) + \dots +  \theta_{n}\delta_{t_{n}}(\mathrm{d}t)\delta_{\mathbf{x}_{n}}(\mathrm{d}\mathbf{x})$ where $\{(\mathbf{x}_{j}, t_{j}): j = 1, \dots, n\}$ is a set of different spatio-temporal locations and $\theta_{j} \in \mathbb{R}$ for $j = 1, \dots, n$. With this specification, $C\{1 \ddagger v(Y)\}$ is the joint cumulant generating function (JCGF) of $Y_{t_{1}}(\mathbf{x}_{1}), \dots, Y_{t_{n}}(\mathbf{x}_{n})$. 
\vspace{2mm}
\begin{Eg} \label{eg:dmeasure}
Let $f(\lambda) = \sum_{k = 1}^{p} q_{k}\delta_{\lambda_{k}}(\lambda)$ for $\lambda_{k}>0$ with $\lambda_{k}\neq\lambda_{k'}$ for $k\neq k'$, $q_{k}>0$ and $p\in\mathbb{N}$ such that $\sum_{k = 1}^{p} q_{k} = 1$. This corresponds to a discrete probability measure for $\lambda$. By substituting the form of $f(\lambda)$ in (\ref{eqn:MSTOULevy}), we find that the JCGF of the resulting MSTOU process is equal to:
\begin{align*}
C\{1\ddagger v(Y)\} &= \sum_{i = 1}^{p} \left(i aq_{k} \int_{\mathbb{R}^{d}\times\mathbb{R}}  h_{A}(\bm{\xi},s, \lambda_{k})\mathrm{d}\bm{\xi}\mathrm{d}s - \frac{1}{2} b q_{k}  \int_{\mathbb{R}^{d}\times\mathbb{R}}   h^{2}_{A}(\bm{\xi},s, \lambda_{k})\mathrm{d}\bm{\xi}\mathrm{d}s \right.\\
&\left.+   \int_{\mathbb{R}^{d}\times\mathbb{R}} \int_{\mathbb{R}} \left(\exp(i h_{A}(\bm{\xi},s, \lambda_{k})z) - 1 - i h_{A}(\bm{\xi},s, \lambda_{k})z\mathbf{1}_{|z|\leq 1}\right) q_{k}\nu(\mathrm{d}z)  \mathrm{d}\bm{\xi}\mathrm{d}s \right). 
\end{align*}
From this expression, we find that the MSTOU process is equal in distribution as the superposition of $p$ independent STOU processes:
\begin{equation*}
\sum_{k = 1}^{p} \int_{A_{t}(\mathbf{x})} \exp(-\lambda_{k}(t-s)) L^{(k)}(\mathrm{d}\bm{\xi}, \mathrm{d}s),
\end{equation*}
where $(L^{k})_{k = 1, \dots, p}$ are independent homogeneous L\'evy bases with characteristic triplets $(q_{k}a, q_{k}b, q_{k}\nu(\mathrm{d}z, \cdot))$. We note that the STOU processes have the same ambit set but different rate parameters and possibly different characteristic triplets of their L\'evy bases.  
\end{Eg}
\vspace{2mm}
\begin{Def}[Spatio-temporal stationarity] \hfill \\
Let $x_{1}, ..., x_{n} \in \mathbb{R}^{d}$ and $t_{1}, ..., t_{n} \in \mathbb{R}$ for $n \in \mathbb{N}$. The spatio-temporal random field $Y_{t}(\mathbf{x})$ is \textit{stationary in space-time} if the joint distribution of $Y_{t_{1}}(\mathbf{x}_{1}), ..., Y_{t_{n}}(\mathbf{x}_{n})$ is the same as that of $Y_{t_{1} + \epsilon}(\mathbf{x}_{1} + \mathbf{u}), ..., Y_{t_{n} + \epsilon}(\mathbf{x}_{n} + \mathbf{u})$ for $\mathbf{u} \in \mathbb{R}^{d}$ and $\epsilon \in \mathbb{R}$.
\end{Def}
\vspace{2mm}
\begin{thm} 
\label{thm:tsstation}
Let $Y_{t}(\mathbf{x})$ be an MSTOU process. Then $Y_{t}(\mathbf{x})$ is stationary in space-time.
\end{thm}
Since $Y$ is stationary, its expectation is the same across space-time locations and the covariance between the process at two locations can be written as a function of their distances apart in space and time:
\\
\begin{Cor} \label{Cor:meancov}
Let $Y$ be an MSTOU process defined by (\ref{eqn:MSTOU}) and $(a, b, \nu(\mathrm{d}z), \mathrm{d}\bm{\xi}\mathrm{d}s f(\lambda)\mathrm{d}\lambda)$ be the CQ of its L\'evy basis $L$. Then, the mean and spatio-temporal covariance of $Y$ are given by: 
\begin{align}
\mathbb{E}\left[Y_{t}(\mathbf{x})\right] &= \left[a +  \int_{\mathbb{R}} z \nu(\mathrm{d}z)\right] \int_{0}^{\infty} \int_{A_{t}(\mathbf{x})} \exp(-\lambda(t-s))\mathrm{d}\bm{\xi}\mathrm{d}s f(\lambda) \mathrm{d}\lambda \nonumber \\
&= \mathbb{E}\left[L'\right] \int_{0}^{\infty} \int_{A_{t}(\mathbf{x})} \exp(-\lambda(t-s))\mathrm{d}\bm{\xi}\mathrm{d}s f(\lambda) \mathrm{d}\lambda, \nonumber \\
\text{and }\Cov(Y_{t}(\mathbf{x}), Y_{t+ d_{t}}(\mathbf{x}+d_{\mathbf{x}})) &= \left[b +  \int_{\mathbb{R}} z^{2}\nu(\mathrm{d}z) \right] \int_{0}^{\infty} \int_{A_{t}(\mathbf{x})\cap A_{t + d_{t}}(\mathbf{x} + d_{\mathbf{x}})} \exp(-2\lambda(t-s) - \lambda d_{t})\mathrm{d}\bm{\xi}\mathrm{d}s f(\lambda) \mathrm{d}\lambda \nonumber \\
&= \Var(L') \int_{0}^{\infty} \int_{A_{t}(\mathbf{x})\cap A_{t + d_{t}}(\mathbf{x} + d_{\mathbf{x}})} \exp(-2\lambda(t-s) -\lambda d_{t})\mathrm{d}\bm{\xi}\mathrm{d}s f(\lambda) \mathrm{d}\lambda, \label{eqn:Ycov}
\end{align} 
where $d_{t} \in \mathbb{R}$ and $d_{\mathbf{x}}\in\mathbb{R}^{d}$ denote distances in time and space while $L'$ denotes the L\'evy seed defined in Definition \ref{def:Lseed} (for our MSTOU process, this does not depend on $\mathbf{z}\in S$).
\end{Cor}
\vspace{2mm} 
\begin{Rem}
From (\ref{eqn:Ycov}), we find that the correlation of $Y$:
\begin{equation*}
\Corr(Y_{t}(\mathbf{x}), Y_{t+ d_{t}}(\mathbf{x}+d_{\mathbf{x}})) = \frac{\int_{0}^{\infty} \int_{A_{t}(\mathbf{x})\cap A_{t + d_{t}}(\mathbf{x} + d_{\mathbf{x}})} \exp(-2\lambda(t-s) -\lambda d_{t})\mathrm{d}\bm{\xi}\mathrm{d}s f(\lambda) \mathrm{d}\lambda}{\int_{0}^{\infty} \int_{A_{t}(\mathbf{x})} \exp(-2\lambda(t-s))\mathrm{d}\bm{\xi}\mathrm{d}s f(\lambda) \mathrm{d}\lambda}. 
\end{equation*}
This means that it depends both on the shape of the integration set $A_{t}(\mathbf{x})$ and $f(\lambda)$. The additional dependence on the ambit set means that it is harder to establish long memory based on regularly varying characteristics as done in \cite{FK2007} and \cite{STW2015}. Nevertheless, we show in Section \ref{sec:gclass} that long memory can be established in our isotropic $g$-class of MSTOU processes.
\end{Rem}

\subsection{Mixing properties}

Spatio-temporal stationarity and mixing properties are useful properties to establish the consistency of moment-based estimators such as the GMM estimators which we construct in Section \ref{sec:Infer}. The following definition is adapted from \cite{PV2017}:
\vspace{2mm}
\begin{Def}[Mixing] \hfill \\
Let $\{Y_{t}(\mathbf{x})\}_{\mathbf{x}\in\mathbb{R}^{d}, t\in\mathbb{R}}$ be a stationary process and $(\mathbf{v}_{n})_{n\in\mathbb{N}}$ be a sequence of spatio-temporal lags such that $\lim_{n\rightarrow\infty} ||\mathbf{v}_{n}||_{\infty}  = \infty$ where $||\cdot||_{\infty}$ refers to the supremum norm. We define the transformation $\theta_{\mathbf{v}}(B)$ such that $\theta_{\mathbf{v}}(B) = \{\omega'\in\Omega: Y_{0}(\mathbf{0})(\omega') = (Y_{0}(\mathbf{0}) + \mathbf{v})(\omega) \text{ for }\omega\in B\}$ for any $B\in\sigma_{Y}$, the $\sigma$-algebra generated by $\{Y_{t}(\mathbf{x})\}$. We call $\{Y_{t}(\mathbf{x})\}$ \textit{mixing} if, for all $A, B\in\sigma_{Y}$:
\begin{equation*}
\lim_{n\rightarrow\infty} P(A\cap\theta_{\mathbf{v}_{n}}(B)) = P(A)P(B). 
\end{equation*}
\end{Def}
\vspace{2mm}
The next result corresponds to the one-dimensional case in Theorem 3.6.~of \cite{PV2017}:
\vspace{2mm}
\begin{thm}\label{thm:mixing}
Let $Y_{t}(\mathbf{x})$ be an MSTOU process. Then, $Y$ is mixing. 
\end{thm}

\subsection{Isotropy and long memory in the $g$-class} \label{sec:gclass}

In this subsection, we look at a class of isotropic MSTOU processes and explore the long-range dependence structures that they can generate. 
\\
\begin{Def}[$g$-class processes] \hfill \\
Let $t\in\mathbb{R}$ and $\mathbf{x}\in\mathbb{R}^{d}$ for $d\in\mathbb{N}$. The \textit{$g$-class of MSTOU processes} is the set of MSTOU processes where the ambit sets are given by:
\begin{equation*}
A_{t}(\mathbf{x}) = \{(\bm{\xi}, s): s\leq t, |\mathbf{x} - \bm{\xi}| \leq g(|t-s|)\},
\end{equation*}
for some positive and non-decreasing function $g:[0, \infty)\rightarrow \mathbb{R}$. 
\end{Def}
\vspace{2mm}
Figure \ref{fig:gclass} shows the ambit sets for $g(|t-s|) = c|t-s|$ for some $c>0$ when we have $d = 1, 2$ and $3$. Due to the exponential kernel in the MSTOU integral, the phenomena observed at a spatial location $\mathbf{x}\in\mathbb{R}^{d}$ is generally more affected by recent events at nearby locations and less affected by older events at locations further away. However, due to the random rate parameters, different events have different levels of influence. By looking at the temporal cross-sections of the ambit sets (i.e.~the spatial ranges for fixed $s$) when $s$ increases towards $t$, we see that the news from or effects of surrounding locations travel towards the point of interest. Here, the parameter $c$ is related to the speed of this travel. For $d = 1$, it determines the length of the spatial line of influence from past epochs; for $d = 2$, it determines the radius of the circle of influence; and for $d = 3$, it determines the radius of the sphere of influence. Similar interpretations hold for more general $g$ functions since they are non-decreasing. For example, we can modulate the behaviour of the travel by setting $g$ to be a quadratic function.
\vspace{2mm}
\begin{Cor} \label{cor:gcon}
Under Assumptions \ref{assum:LCQ} and \ref{assum:finite2m}, a $g$-class MSTOU process is well-defined if:
\begin{equation}
\int_{0}^{\infty} \int_{-\infty}^{t} g^{d}(|t-s|) \exp(-\lambda(t-s)) f(\lambda) \mathrm{d}s \mathrm{d}\lambda < \infty \text{ and } \int_{0}^{\infty} \int_{-\infty}^{t}  g^{d}(|t-s|)\exp(-2\lambda(t-s)) f(\lambda) \mathrm{d}s \mathrm{d}\lambda < \infty. \label{eqn:gcon}
\end{equation} 
\end{Cor}
\vspace{2mm}
\begin{Eg}\label{eg:gcanon}
Consider the case with $g(|t-s|) = c|t-s|$ for some $c>0$. Then, (\ref{eqn:gcon}) holds when $\int_{0}^{\infty} \frac{1}{\lambda^{d+1}} f(\lambda) \mathrm{d}\lambda < \infty$. This is fulfilled for example when $f(\lambda) = \frac{\beta^{\alpha}}{\Gamma(\alpha)} \lambda^{\alpha-1} e^{-\beta \lambda}$, the Gamma density with shape and rate parameters,  $\alpha>d+1$ and $\beta >0$ since:
\begin{equation*}
\int_{0}^{\infty}\frac{1}{\lambda^{d+1}} f(\lambda) \mathrm{d}\lambda = \frac{\beta^{d+1} }{(\alpha - 1)\dots(\alpha - (d+1))} < \infty.
\end{equation*}
\end{Eg} 
\vspace{2mm}
Now that we have simple integrability conditions for the $g$-class, we proceed to prove its key property: isotropy. 
\vspace{2mm}
\begin{Def}[Isotropy] \hfill \\
Let $t\in\mathbb{R}$ and $\mathbf{x}\in\mathbb{R}^{d}$. A spatio-temporal process $Y_{t}(\mathbf{x})$ is called \textit{isotropic} if its spatial covariance:
\begin{align*}
\Cov(Y_{t}(\mathbf{x}), Y_{t}(\mathbf{x}+d_{\mathbf{x}})) &= C(|d_{\mathbf{x}}|),
\end{align*}
for some positive definite function $C$.
\end{Def}
\vspace{2mm}
\begin{thm} \label{thm:giso}
Let $Y$ be a $g$-class MSTOU process. Then, $Y$ is isotropic in space.
\end{thm}
\vspace{2mm}

\begin{figure}[tbp]
\centering
\caption{Ambit sets in one to three spatial dimensions for a $g$-class MTOU process with $g(|t-s|) = c|t-s|$ for some $c>0$.}
\label{fig:gclass}
\includegraphics[width = 6.8in, height = 3in, trim = 1in 0.4in 0.4in 0in]{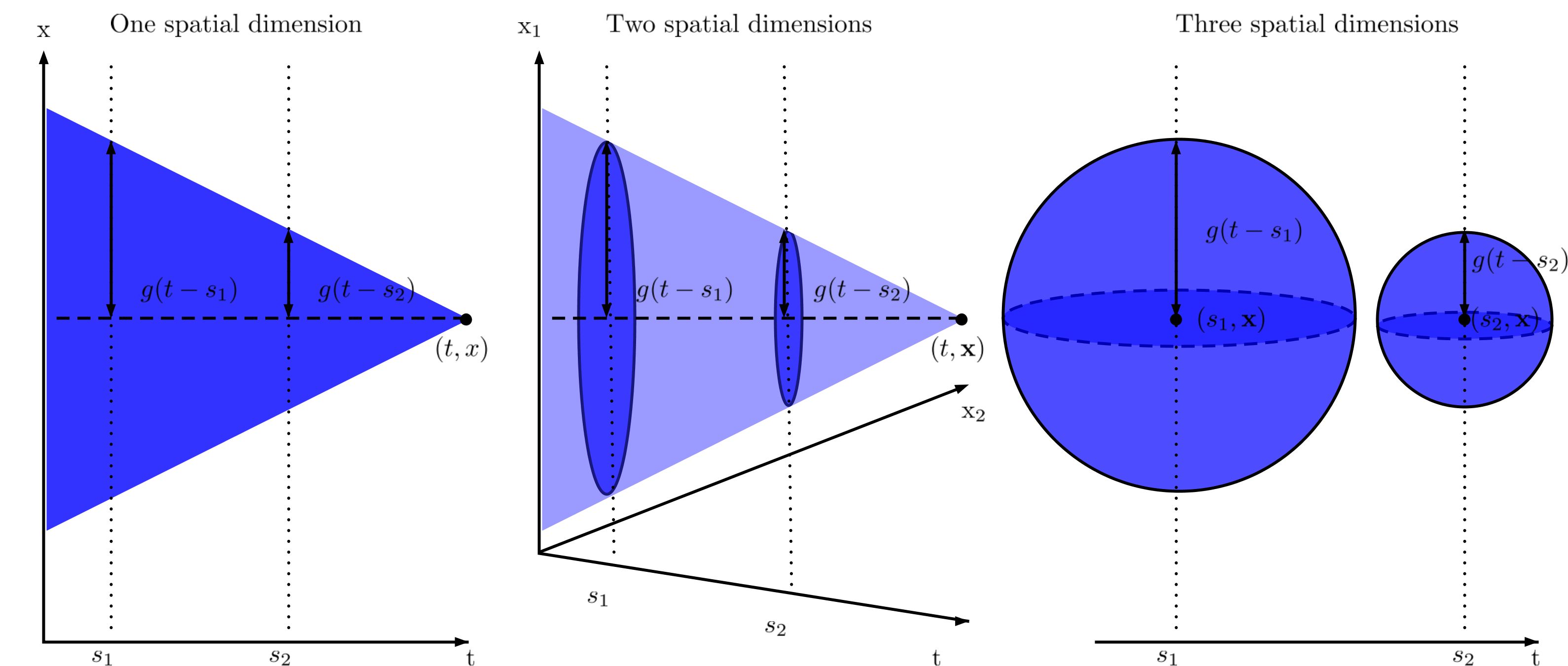}
\end{figure}

\begin{Rem}\label{rem:23D}
In general, the spatial covariance takes a simpler form when $d = 3$ as compared to $d = 2$. This is because we replace $(2g(|t-s|)- d_{x})$ in (\ref{eqn:scovg}) with $\pi(4g(|t-s|) + |d_{x}|)(2g(|t-s|) - |d_{x}|)^{2}/12$ when $d = 3$ which is of a simpler functional form than the $2g^{2}(|t-s|)\cos^{-1}(|d_{\mathbf{x}}|/2g(|t-s|)) - (|d_{\mathbf{x}}|\sqrt{4g^{2}(|t-s|) - |d_{\mathbf{x}}|^{2}})/2$ used for $d = 2$. For mathematical simplicity, it may be useful to embed data in two spatial dimensions in a modelling scenario with three spatial dimensions. However, checks should be made to ensure that if the assumptions on the additional dimension are reasonable for the context.
\end{Rem}
\vspace{2mm}
Examples \ref{eg:canon1d} and \ref{eg:canon1dprop} show that we can construct non-separable covariances using the $g$-class. This means that the spatio-temporal covariances cannot be expressed as products of spatial and temporal covariances \cite[]{CW2011}. In what follows, we also obtain explicit expressions which are useful for inference.
\vspace{2mm}
\begin{Eg}\label{eg:canon1d}
Consider the scenario in Example \ref{eg:gcanon} with $d = 1$. From Corollary \ref{Cor:meancov}, we have that:
\begin{align}
\mathbb{E}\left[Y_{t}(x)\right] &= \mathbb{E}\left[L'\right] \int_{0}^{\infty} \int_{A_{t}(x)} \exp(-\lambda(t-s))\mathrm{d}\xi\mathrm{d}s f(\lambda) \mathrm{d}\lambda, \nonumber \\
&= \mathbb{E}\left[L'\right] \int_{0}^{\infty}\frac{2c}{\lambda^{2}} f(\lambda) \mathrm{d}\lambda, \nonumber \\
&= \frac{2c \beta^{2}\mathbb{E}\left[L'\right]}{(\alpha - 2)(\alpha - 1)} \int_{0}^{\infty}  \frac{\beta^{\alpha-2}}{\Gamma(\alpha-2)} \lambda^{\alpha-2-1} e^{-\beta \lambda}\mathrm{d}\lambda\nonumber  \\
&= \frac{2c \beta^{2}\mathbb{E}\left[L'\right] }{(\alpha - 2)(\alpha - 1)}, \nonumber \\
\text{and } \Cov(Y_{t}(x), Y_{t+ d_{t}}(x+d_{x})) &= \Var(L') \int_{0}^{\infty} \int_{A_{t}(x)\cap A_{t + d_{t}}(x + d_{x})} \exp(-2\lambda(t-s) -\lambda d_{t})\mathrm{d}\xi\mathrm{d}s f(\lambda) \mathrm{d}\lambda, \nonumber \\
 &= \Var(L')  \int_{0}^{\infty} \frac{c}{2\lambda^{2}}\exp\left(-\lambda \max(|d_{t}|, |d_{x}|/c)\right) f(\lambda) \mathrm{d}\lambda \label{eqn:intint} \\
&=  \Var(L')\frac{c\beta^{\alpha}}{2(\beta+A)^{\alpha - 2}(\alpha - 2)(\alpha - 1)} \int_{0}^{\infty} \frac{(\beta + A)^{\alpha-2}}{\Gamma(\alpha - 2)} \lambda^{\alpha-2-1} e^{-(\beta+A)\lambda} \mathrm{d}\lambda \nonumber \\
&=  \Var(L')\frac{c\beta^{\alpha}}{2(\beta+A)^{\alpha - 2}(\alpha - 2)(\alpha - 1)}, \nonumber\\
&=\frac{c\beta^{\alpha}\Var(L')}{2(\beta+\max(|d_{t}|, |d_{x}|/c))^{\alpha - 2}(\alpha - 2)(\alpha - 1)}, \nonumber
\end{align}
where $A = \max(|d_{t}|, |d_{x}|/c)$ and (\ref{eqn:intint}) holds from the results in Example 3 of \cite{NV2016}.
\end{Eg}
\vspace{2mm}
\begin{Eg}\label{eg:canon1dprop}
Let $L$ be a spatio-temporal extension of the compound Poisson L\'evy basis defined in \cite{FK2007}:
\begin{equation*}
L(E) = \sum_{k = -\infty}^{\infty} Z_{k}\mathbf{1}_{\{(\Gamma_{k}, \lambda_{k})\in A\}} \text{ for } E\in\mathcal{B}_{b}(S),
\end{equation*} 
where $\{Z_{k}\}_{k\in\mathbb{N}}$ is a sequence of independent, identically distributed (i.i.d.) random variables with distribution function $G$, $\left\{\Gamma_{k} = \left(\Gamma_{k}^{(1)}, \Gamma_{k}^{(2)}\right)\right\}_{k\in\mathbb{N}}$ denote the spatio-temporal jump locations of a Poisson process $N = (N_{t}(\mathbf{x}))_{(\mathbf{x}, t)\in \mathbb{R}^{d}\times\mathbb{R}}$ with intensity $\mu$, and $\lambda_{k}$ is an i.i.d. sequence with probability density function $f$. These three components are independent of each other.
\\
This means that the L\'evy seed is a compound Poisson random variable, i.e.~$L' = \sum_{k = 1}^{N_{1}(\mathbf{1})}Z_{k}$. Its mean and variance are:
\begin{equation*}
\mathbb{E}\left[L'\right] = \mu\mathbb{E}\left[Z_{k}\right] \text{ and } \Var(L') = \mu\left(\Var(Z_{k}) + \left(\mathbb{E}\left[Z_{k}\right]\right)^{2}\right).
\end{equation*}
Suppose that for the model in Example \ref{eg:canon1d}, we have $Z_{k} \sim$ Gamma$(\alpha_{Z}, \beta_{Z})$ for $k \in\mathbb{N}$ then:
\begin{align*}
\mathbb{E}\left[Y_{t}(x)\right] &= \frac{2c \beta^{2}\mu\alpha_{Z}}{(\alpha - 2)(\alpha - 1)\beta_{Z}},\\
\text{and } \Cov(Y_{t}(x), Y_{t+ d_{t}}(x+d_{x})) &=\frac{c\beta^{\alpha}\mu \alpha_{Z}(\alpha_{Z}+1)}{2(\beta+\max(|d_{t}|, |d_{x}|/c))^{\alpha - 2}(\alpha - 2)(\alpha - 1)\beta_{Z}^{2}}. 
\end{align*}
\end{Eg}
\vspace{2mm}
To compute the spatio-temporal covariance when $g(|t-s|) = c|t-s|$ for $c>0$ and $d>1$, we need to consider $A_{t}(\mathbf{x})\cap A_{t}(\mathbf{x}+d_{\mathbf{x}})$ in two cases: $|d_{\mathbf{x}}| > c d_{t}$ and  $|d_{\mathbf{x}}| \leq c d_{t}$ for $d_{t}\geq 0$. For the former case, the intersection begins at time $t^{*} = t + (d_{t} - |d_{\mathbf{x}}|/c)/2$ and the temporal cross-section is equal to the volume of the intersection of two $d$-spheres with centres $\mathbf{x}$ and $\mathbf{x} + d_{\mathbf{x}}$, and radii $g(|t-s|)$ and $g(|t+d_{t} -s|)$ respectively; for the latter case, the intersection begins at $t$ and the temporal cross-section is equal to the volume of the $d$-sphere with centre $\mathbf{x}$ and radius $g(|t-s|)$. 
\\
It is more complicated to work out the spatio-temporal covariances for a general $g$ function because the forms of $A_{t}(\mathbf{x})\cap A_{t}(\mathbf{x}+d_{\mathbf{x}})$ would depend on the curvature of the ambit set. Instead of computing spatio-temporal covariances, we now focus on the spatial and the temporal covariances separately in order to establish short-range or long-range dependence. 
\vspace{2mm}
\begin{Def}[Temporal and spatial short/long-range dependence] \hfill \\
The spatio-temporal process $\{Y_{t}(\mathbf{x}): t\in\mathbb{R}, \mathbf{x}\in\mathbb{R}^{d}\}$ is said to have \textit{temporal short-range dependence} if:
\begin{equation*}
\int_{0}^{\infty} \Cov(Y_{t}(\mathbf{x}), Y_{t+ \tau}(\mathbf{x}))\mathrm{d}\tau < \infty,
\end{equation*}
and \textit{temporal long-range dependence} if the integral is infinite. 
\\
Similarly, an isotropic process has \textit{spatial short-range dependence} if:
\begin{equation*}
\int_{0}^{\infty} C(r)\mathrm{d}r < \infty,
\end{equation*}
where $\Cov(Y_{t}(\mathbf{x}), Y_{t}(\mathbf{x}+d_{\mathbf{x}})) = C(|d_{\mathbf{x}}|)$ and $r = |d_{\mathbf{x}}|$. It is said to have \textit{spatial long-range dependence} if the integral is infinite.
\end{Def}
\vspace{2mm}
\begin{Eg}
Consider the model used in Example \ref{eg:canon1d}. Set $r = d_{x} = 0$ and $\tau = d_{t}$. Then:
\begin{align*}
\int_{0}^{\infty} \Cov(Y_{t}(x), Y_{t+ \tau}(x))\mathrm{d}\tau &= \frac{c\beta^{\alpha}\Var(L')}{2(\alpha - 2)(\alpha - 1)} \int_{0}^{\infty} (\beta+\tau)^{-(\alpha - 2)} \mathrm{d}\tau \\
&= \frac{c\beta^{\alpha}\Var(L')}{2(\alpha - 2)(\alpha - 1)} \left[ \frac{(\beta+\tau)^{-(\alpha - 3)}}{3 - \alpha} \right]_{0}^{\infty} \\
&= \frac{c\beta^{3}\Var(L')}{2(\alpha - 2)(\alpha - 1)(\alpha - 3)}, 
\end{align*}
for $\alpha >3$. For $2 < \alpha \leq 3$, this integral is infinite and the process has temporal long-range dependence. These parameter bounds also apply to spatial long-range dependence since if $r = d_{x}$ and $\tau = d_{t} = 0$, we have:
\begin{align*}
\int_{0}^{\infty} C(r)\mathrm{d}r &= \frac{c\beta^{\alpha}\Var(L')}{2(\alpha - 2)(\alpha - 1)} \int_{0}^{\infty} (\beta+ r/c)^{-(\alpha - 2)} \mathrm{d}r\\
&= \frac{c\beta^{\alpha}\Var(L')}{2(\alpha - 2)(\alpha - 1)} \left[ \frac{c(\beta+r/c)^{-(\alpha - 3)}}{3 - \alpha} \right]_{0}^{\infty} \\
&= \frac{c^{2}\beta^{3}\Var(L')}{2(\alpha - 2)(\alpha - 1)(\alpha - 3)}, 
\end{align*}
for $\alpha > 3$. But the integral diverges for $2<\alpha\leq 3$. 

\begin{figure}[tbp]
\centering
\caption{(a) Three choices of $f(\lambda)$ and (b) the spatial correlation structures ($\rho^{(S)}$) of the corresponding MSTOU processes. Since we have set $c = 1$, these share the same forms as the temporal correlations.}
\label{fig:fcorr}
\includegraphics[width = 4.5in, height = 2.2in, trim = 0.4in 0.4in 0.4in 0in]{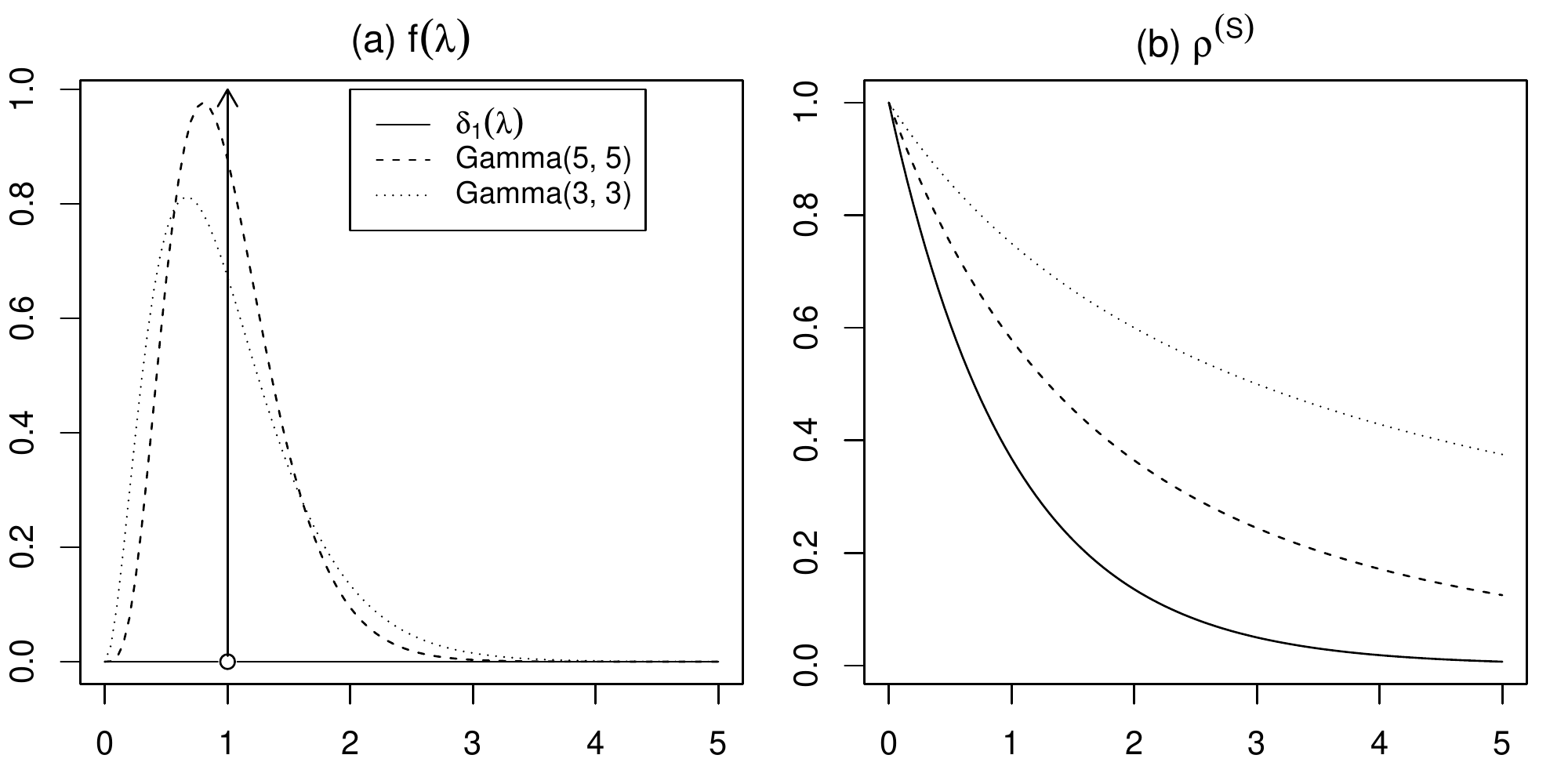}
\end{figure}

Figure \ref{fig:fcorr} shows three choices of $f(\lambda)$ and the spatial correlation structures of the corresponding MSTOU processes. Here, we have the results for the Dirac delta measure at $1$ in bold curves, that for the Gamma$(5, 5)$ density in dashed curves and that for the Gamma$(3, 3)$ density in dotted curves. For all the cases, we have set $c = 1$ so that the temporal correlation function is the same as the spatial one. While the first case leads to exponential correlation, the second and third lead to short-range and long-range correlation respectively. 
\end{Eg}
\vspace{2mm}
We can establish long-range dependence for a similar process in three spatial dimensions:
\vspace{2mm}
\begin{Eg}\label{eg:3Dcanon}
Consider three dimensional space ($d = 3$) and the case with $g(|t-s|) = c|t-s|$ for $c>0$. Let $f(\lambda)$ be the Gamma$(\alpha, \beta)$ density with $\beta > 0$ and $\alpha >4$. From the proof of Theorem \ref{thm:giso} and Remark \ref{rem:23D}, we have that the spatial covariance of our process is:
\begin{align}
\Cov(Y_{t}(\mathbf{x}), Y_{t}(\mathbf{x}+|d_{\mathbf{x}}|)) &= \frac{\pi\Var(L')}{12} \int_{0}^{\infty} \int_{|d_{\mathbf{x}}|/2c}^{\infty}(4cw + |d_{\mathbf{x}}|)(2cw - |d_{\mathbf{x}}|)^{2}\exp(-2\lambda w)\mathrm{d}w f(\lambda) \mathrm{d}\lambda \nonumber \\
&= \frac{c^{2}\pi\Var(L')}{4} \int_{0}^{\infty}\frac{(\lambda |d_{\mathbf{x}}| + 2c)e^{-\lambda |d_{\mathbf{x}}|/c}}{\lambda^{4}} f(\lambda) \mathrm{d}\lambda \nonumber \\
&=  \frac{\beta^{\alpha}c^{\alpha-1}\pi\Var(L')}{4(\alpha - 4)(\alpha - 3)(\alpha - 2)(\alpha - 1)}(\beta c + |d_{\mathbf{x}}|)^{3-\alpha}(2\beta c + (\alpha - 2)|d_{\mathbf{x}}|) \nonumber \\
&=  \frac{\beta^{4}c^{3}\pi\Var(L')}{2(\alpha - 4)(\alpha - 3)(\alpha - 2)(\alpha - 1)} \left(\frac{\beta c + |d_{\mathbf{x}}|}{\beta c}\right)^{3 - \alpha}\left(\frac{2\beta c + (\alpha - 2)|d_{\mathbf{x}}|}{2\beta c}\right) \label{eqn:3Dcov}.
\end{align}
Without loss of generality, let $d_{t} \geq 0$. To compute the temporal covariance, we set $d_{\mathbf{x}} = \mathbf{0}$. Then, $A_{t}(\mathbf{x})\cap A_{t + d_{t}} = A_{t}(\mathbf{x})$. In this case, the temporal cross-section of $A_{t}(\mathbf{x})$ corresponds to a sphere with radius $g(|t-s|)$. So:
\begin{align}
\Cov(Y_{t}(\mathbf{x}), Y_{t + d_{t}}(\mathbf{x}+|d_{\mathbf{x}}|)) &= \frac{4\pi\Var(L')}{3} \int_{0}^{\infty} \int_{0}^{\infty} (cw)^{3}\exp(-2\lambda w - \lambda d_{t})\mathrm{d}w f(\lambda) \mathrm{d}\lambda \nonumber \\
&=  \frac{c^{3}\pi\Var(L')}{2} \int_{0}^{\infty} \frac{1}{\lambda^{4}}\exp(- \lambda d_{t}) f(\lambda) \mathrm{d}\lambda \nonumber \\
&=   \frac{\beta^{4}c^{3}\pi\Var(L')}{2(\alpha - 4)(\alpha - 3)(\alpha - 2)(\alpha - 1)}\left(\frac{\beta}{\beta + d_{t}}\right)^{\alpha - 4}. \label{eqn:3Dtcov}
\end{align}
Using (\ref{eqn:3Dtcov}), we have that: 
\begin{align*}
\int_{0}^{\infty} \Cov(Y_{t}(\mathbf{x}), Y_{t+ \tau}(\mathbf{x}))\mathrm{d}\tau &= \int_{0}^{\infty} \frac{\beta^{4}c^{3}\pi\Var(L')}{2(\alpha - 4)(\alpha - 3)(\alpha - 2)(\alpha - 1)}\left(\frac{\beta}{\beta + \tau}\right)^{\alpha - 4} \mathrm{d}\tau \\
&= \frac{\beta^{\alpha}c^{3}\pi\Var(L')}{2(\alpha - 4)(\alpha - 3)(\alpha - 2)(\alpha - 1)}\left[ \frac{\left(\beta + \tau\right)^{5 - \alpha}}{5 - \alpha} \right]_{0}^{\infty} \\
&= \frac{\beta^{5}c^{3}\pi\Var(L')}{2(\alpha - 5)(\alpha - 4)(\alpha - 3)(\alpha - 2)(\alpha - 1)},
\end{align*}
for $\alpha >5$. But the integral diverges for $4<\alpha\leq 5$. Similarly, using (\ref{eqn:3Dcov}), we have:
\begin{align*}
\int_{0}^{\infty} C(r)\mathrm{d}r &=  \int_{0}^{\infty} \frac{\beta^{4}c^{3}\pi\Var(L')}{2(\alpha - 4)(\alpha - 3)(\alpha - 2)(\alpha - 1)} \left(\frac{\beta c + r}{\beta c}\right)^{3 - \alpha}\left(\frac{2\beta c + (\alpha - 2)r}{2\beta c}\right) \mathrm{d}r \\
&= \frac{3\beta^{5}c^{4}\pi\Var(L')}{4(\alpha -5)(\alpha - 4)(\alpha - 3)(\alpha - 2)(\alpha - 1)},
\end{align*}
for $\alpha > 5$ and the integral diverges for $4<\alpha\leq 5$.
\end{Eg}
\vspace{2mm}

\subsection{Relation to the spatio-temporal CAR$_{\wedge}$ process}

In Example \ref{eg:dmeasure}, we saw that when $f(\lambda)$ is concentrated at $p$ distinct values, MSTOU processes are equal in law to a sum of $p$ independent STOU processes. Here, we consider so-called spatio-temporal CAR$_{\wedge}(p)$ processes and show that they too can be represented as superpositions of $p$ STOU processes. However, these STOU processes share the same underlying L\'evy basis and are correlated. 
\vspace{2mm}
\begin{Def}[Spatio-temporal $\text{CAR}_{\wedge}(p)$ process] \label{def:stCAR}\hfill \\
We call a random field in space-time ($\mathbb{R}^{d}\times \mathbb{R}$) a \textit{spatio-temporal CAR$_{\wedge}(p)$ process} if:
\begin{equation*}
Y_{t}(\mathbf{x}) = \mathbf{b}^{T}\mathbf{X}_{t}(\mathbf{x}),
\end{equation*}
where $\mathbf{b} = (1, 0, \dots ,0)^{T} \in \mathbb{R}^{p}$ and:
\begin{align*}
\mathbf{X}_{t}(\mathbf{x}) &= \int_{A_{t}(\mathbf{x})}\exp\big(A(t-s)\big)\mathbf{e}_{p} L(\mathrm{d}\bm{\xi}, \mathrm{d}s), \\
\text{with } A &= \begin{pmatrix}
0 & 1 & 0 & \dots &  0 \\
0 &  0 & 1 & \dots & 0 \\
\vdots & \vdots & \vdots & \ddots & \vdots \\
-a_{p} & -a_{p-1} & -a_{p-2} & \dots & -a_{1}
\end{pmatrix},
\end{align*}
and $a_{1}, \dots, a_{p} \in\mathbb{R}$. Similar to our ambit set for MSTOU processes, $A_{t}(\mathbf{x}) = (\mathbf{x}, t)$ satisfies the conditions in (\ref{eqn:ambitassumptions}). Here, $L$ is a homogeneous L\'evy basis and $\mathbf{e}_{p}$ is the $p^{\text{th}}$ Euclidean basis vector.
\end{Def}
\vspace{2mm}
\begin{thm}\label{thm:CARp}
Let $Y_{t}(\mathbf{x})$ be a spatio-temporal CAR$_{\wedge}$(p) process as defined in Definition \ref{def:stCAR}. Then:
\begin{equation*}
Y_{t}(\mathbf{x}) = \sum_{i = 1}^{p}\frac{1}{\prod\limits_{\stackrel{1\leq m \leq p}{m\neq i}} (\eta_{i} - \eta_{m})} \int_{A_{t}(\mathbf{x})} \exp(\eta_{i}(t-s)) L(\mathrm{d}\bm{\xi}, \mathrm{d}s),
\end{equation*}
where $\eta_{1}, \dots, \eta_{p}$ are the corresponding negative and distinct eigenvalues of $A$. 
\end{thm}
\vspace{2mm}
The definition of a spatio-temporal CAR$_{\wedge}(p)$ process is a spatio-temporal extension of the usual CAR$(p)$ process where an ambit set is incorporated to allow for non-separable spatio-temporal dependence. When $L$ has finite second moments, integrability conditions similar to (\ref{eqn:simicon}) are required for the process to be well-defined. Just as how the purely temporal case does not result in temporal long-range dependence (see Remark 8 in \cite{FK2007}), this construction does not result in temporal long-range dependence since if the process is well-defined:
\begin{equation*}
\int_{0}^{\infty} \Cov(Y_{t}(\mathbf{x}), Y_{t+ \tau}(\mathbf{x}))\mathrm{d}\tau = \sum_{i = 1}^{p}\frac{\Var(L')}{\left(-\eta_{i}\right)\prod\limits_{\stackrel{1\leq m \leq p}{m\neq i}} (\eta_{i} - \eta_{m})} \int_{A_{t}(\mathbf{x})} \exp(2\eta_{i}(t-s)) \mathrm{d}\bm{\xi} \mathrm{d}s < \infty. 
\end{equation*}
This shows that the ability to model temporal long-range dependence comes from the choice of $f(\lambda)$, the probability density of the rate parameter, rather than on the choice of the ambit set.

\section{Simulation and compound Poisson MSTOU processes} \label{sec:Sim}

In this section, we develop a simulation algorithm for MSTOU processes which involve the compound Poisson L\'evy basis mentioned in Example \ref{eg:canon1dprop}. This is a generalisation and combination of the simulation algorithms in \cite{BM2017} and \cite{FK2007} for compound Poisson continuous auto-regressive moving average (CARMA) random fields on $\mathbb{R}^{d}$  and positive shot noise processes on $\mathbb{R}$ respectively. As such, the processes that we simulate can be seen as spatio-temporal shot-noise processes.
\\
We call $Y_{t}(\mathbf{x})$ a \textit{spatio-temporal shot noise process} if:
\begin{equation}
Y_{t}(\mathbf{x}) = \int_{0}^{\infty}\int_{A_{t}(\mathbf{x})} \exp(-\lambda(t-s)) L(\mathrm{d}\boldsymbol{\xi}, \mathrm{d}s, \mathrm{d}\lambda) = \sum_{k = -\infty}^{\infty} e^{-\lambda_{k}\left(t - \Gamma_{k}^{(2)}\right)} Z_{k}\mathbf{1}_{A_{t}(\mathbf{x})}(\Gamma_{k}).
\label{eqn:psn}
\end{equation}
Similar to the approach in \cite{BM2017} for CARMA random fields, we simulate our MSTOU process over a bounded space-time region $D$. This means that we approximate (\ref{eqn:psn}) by:
\begin{equation}
Z_{t}(\mathbf{x}) = \sum_{k = 1}^{M} e^{-\lambda_{k}\left(t - \Gamma_{k}^{(2)}\right)} Z_{k}\mathbf{1}_{A_{t}(\mathbf{x})}(\Gamma_{k}).
\label{eqn:psnapprox}
\end{equation}
Here, $M$ denotes the number of jumps in $D$ and $M\sim$ Poisson$(\mu\Leb(D))$. The $M$ jump locations are uniformly distributed about $D$. 
\\
Let $\{(\mathbf{x}_{i}, t_{j}) : i = 1, \dots, n \text{ and } j = 1, \dots, m\}$ denote a spatio-temporal grid in $D$. Algorithm \ref{alg:psn} shows how we can simulate our MSTOU process in the case of one-dimensional space based on this approximation when $f$ is the Gamma$(\alpha, \beta)$ density and $Z_{k}\sim \Gamma(\alpha_{Z}, \beta_{Z})$. To extend this to $d$-dimensional space, we need to use arrays to store the process values instead of a matrix and extend the `for' loop operations. 
\\
To reduce kernel truncation error, we should pad the boundaries of $D$ and implement our algorithm on an extended domain. Since our ambit set $A_{t}(\mathbf{x})$ does not include times after $t$, we only need to pad the temporal domain from the past. The extent of the padding and its effectiveness depends on the smallest generated value of $\lambda_{k}$: the smaller the minimum value of $\lambda_{k}$, the wider our extended domain needs to be. A good indicator to monitor would be $\exp(-\lambda_{\min}T_{pad})$ where $\lambda_{min}>0$ and $T_{pad}>0$ denote the minimum $\lambda_{k}$ and the time padding respectively. 
\\
We note that unlike the discrete convolution algorithms for STOU processes in \cite{NV2016}, we do not have kernel discretisation error since we only evaluate the exponential kernel at jump locations. All our simulation error is due to the kernel truncation imposed by the extent of the padding. There is also no ambit set approximation error except that related to the kernel truncation. This means that we are free to choose a grid size for our simulation domain based on our needs. If we want to simulate processes with longer memory than that corresponding to exponential correlations and estimate from our results, we require data over large areas. Thus, we might want to choose large grid sizes in order to cover a large area in a reasonable computational time. On the other hand, if we are not interested in estimating our simulated data, we can choose finer simulation grids over smaller domains.
\\
Plot (a) in Figure \ref{fig:psn} illustrates the simulated jumps in the extended domain of $[-40, 140]\times[-40, 100]$ for a canonical spatio-temporal shot noise process, i.e.~the case in one-dimensional space with $A_{t}(x) = \{(\xi, s): |x-\xi|<c|t-s|\}$. Here, we have padded the original simulation domain of $[0, 100]\times[0, 100]$ by $40$ units in both spatial directions and in the direction towards the past. The rate parameter of the underlying Poisson process is $\mu = 0.2$ and results in $4995$ jumps over the extended domain. In this case, we have $\exp(-\lambda_{min}T_{pad}) = 0.000725$ so the kernel truncation error should be quite small. The other model parameters are $\alpha = \alpha_{z} = 3$ and $\beta = \beta_{z} = 1$. In order to cover $[0, 100]\times[0, 100]$ in a reasonable amount of time, we choose a grid size of $\triangle = 0.5$. Plots (b) and (c) show the heat and perspective plots of the corresponding simulation. It is interesting to see that the linear edges of the ambit set are also reflected in the heat plot  

\clearpage

\begin{algorithm}[tbp]
\caption{Simulating a space-time positive shot noise process over a bounded domain with one-dimensional space.}\label{alg:psn}
\begin{algorithmic}[1]
\State $M \gets rpois(1, \mu\Leb(D))$ \Comment{Generate the number of jumps from a Poisson distribution.}
\State $\Gamma \gets runif(M, D)$ \Comment{Generate $M$ spatio-temporal jump locations from a Uniform distribution over $D$.}
\State $\Lambda \gets rgamma(M, \alpha, \beta)$ \Comment{Generate $M$ rate parameters from a Gamma$(\alpha, \beta)$ distribution.}
\State $Z \gets rgamma(M, \alpha_{Z}, \beta_{Z})$ \Comment{Generate $M$ jump values from a Gamma$(\alpha_{Z}, \beta_{Z})$ distribution.}
\State $Y \gets matrix(0, n, m)$ \Comment{Create a storage matrix for our simulated data.}
\For{$i = 1, \dots, n$} 
\For {$j = 1, \dots, m$}
\State $y \gets 0$ \Comment{Create a variable for the $i-j$th entry of $Y$.}
\For{$k = 1, \dots, M$}
\If{$\Gamma_{k} \in A_{t_{j}}(\mathbf{x}_{j})$}
\State $y \gets y + e^{-\lambda_{k}(t_{j} - \Gamma_{k}^{(2)})}Z_{k}$ \Comment{Add contribution of $M$th jump to the process value if it lies in $A_{t}(\mathbf{x})$.}
\EndIf
\EndFor
\State $Y[i, j] \gets y$ \Comment{Store the final process value for the $i-j$th location.}
\EndFor
\EndFor
\end{algorithmic}
\end{algorithm}

\begin{figure}[tbp]
\centering
\caption{Simulating a canonical spatio-temporal shot noise process with $c = 1$: (a) jumps in extended domain $[-40, 140]\times[-40, 100]$ (in red: jumps in the simulation domain); (b) heat plot over $[0, 100] \times [0, 100]$ where the values are generated with a grid spacing of $0.5$ units; (c) perspective plot of the same realisation. The parameter values for the Gamma jump and rate parameter distributions are $\alpha = \alpha_{z} = 3$ and $\beta = \beta_{z} = 1$, while the rate parameter for the underlying Poisson process is $\mu = 0.2$.}
\label{fig:psn}
\includegraphics[width = 6in, height = 2.4in, trim = 1in 0.4in 1in 0in]{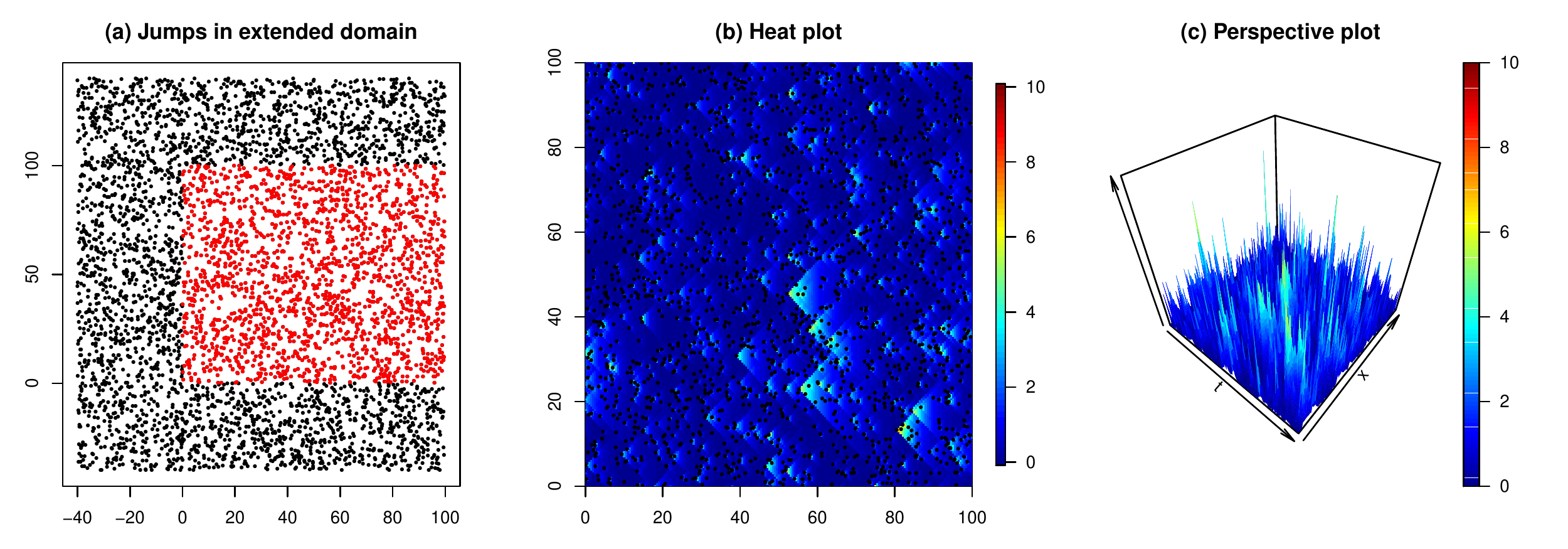}
\end{figure}

\begin{figure}[tbp]
\centering
\caption{Heat plots over $[0, 25]\times[0, 25]$ of: (a) the MSTOU process and (b) the corresponding STOU process with rate parameter $\int_{0}^{\infty}\lambda f(\lambda) \mathrm{d}\lambda = \alpha/\beta = 3$. The black dots denote the positions of the jumps in space-time.}
\label{fig:psnstou}
\includegraphics[width = 4.6in, height = 2.6in, trim = 1in 0.4in 0.4in 0in]{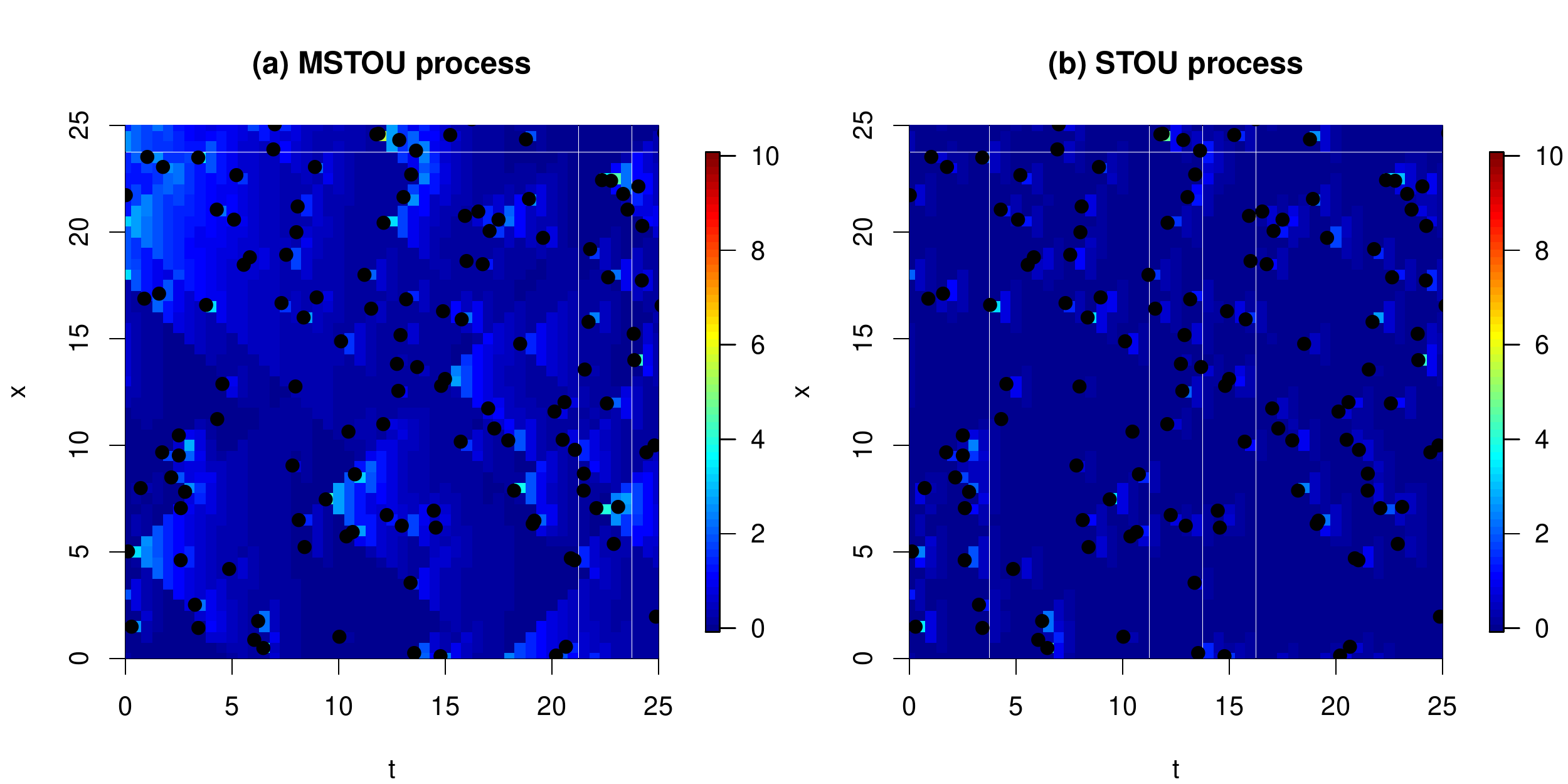}
\end{figure}

\clearpage

because they determine which jumps affect the process value at a particular location.
\\
For a better understanding of how an MSTOU process works, we zoom into our heat plot and compare our process to a STOU process with the same parameter settings but with its rate parameter set to $\int_{0}^{\infty}\lambda f(\lambda) \mathrm{d}\lambda $. From Figure \ref{fig:psnstou}, we see that the jumps in both processes typically occur at the Poisson jump locations which are denoted by the black dots. However, while the values decay at the same rate for the STOU process in Plot (b), the values decay at varying rates for each jump in the MSTOU process. The lower rate parameters lead to larger clusters which is consistent with the long memory of the process.  
\\
Figure \ref{fig:TSACF2} shows the series and autocorrelation (ACF) plots for the simulated process at a fixed spatial location ($x = 100$) and a fixed temporal location ($t = 100$). The black curves in the ACF plots denote the theoretical correlations. We see that our simulation replicates the dependence structure of the process quite well although there is some discrepancies at higher lags which are possibly due to simulation error, random variation and the lower amount of data for estimation. 
\\
For the $g$-class, we calculate an upper bound for the mean squared error (MSE) as follows:
\vspace{2mm}
\begin{thm} \label{thm:MSE}
Let $\{Y_{t}(\mathbf{x})\}_{\mathbf{x}\in\mathbb{R}^{d}, t\in\mathbb{R}}$ be a spatio-temporal shot noise process in the $g$-class and let $Z_{t}(\mathbf{x})$ be its simulation approximation given by (\ref{eqn:psnapprox}). Then:
\begin{align}
\mathbb{E}\left[\left(Y_{t}(\mathbf{x}) - Z_{t}(\mathbf{x})\right)^{2}\right] &\leq \frac{\pi^{d/2}\left(\Var(L') + \mathbb{E}\left[L'\right]^{2}\right)}{\Gamma\left(\frac{d}{2} + 1\right)} \int_{0}^{\infty}\left(\int_{\min(T_{pad}, g^{-1}(X_{pad}))}^{\infty}  g^{d}(w)e^{-2\lambda w}\mathrm{d}w\right) f(\lambda) \mathrm{d}\lambda, \label{eqn:MSE}
\end{align}
where $X_{pad}>0$ is the space padding in the simulation.
\end{thm}
\vspace{2mm}
The MSE upper bound (\ref{eqn:MSE}) shrinks to zero as the padding extents $T_{pad}, X_{pad} \rightarrow \infty$ as we now show for a particular case:
\vspace{2mm}
\begin{Eg}
Consider the case when $g(|t-s|) = c|t-s|$ for $c>0$ and $d = 1$. Let $f(\lambda)$ be the Gamma$(\alpha, \beta)$ density with $\alpha>2$ and $\beta>0$. Then, the upper bound on the MSE is given by:
\begin{align}
(\ref{eqn:MSE}) &=  c\left(\Var(L') + \mathbb{E}\left[L'\right]^{2}\right) \int_{0}^{\infty}\frac{1}{\lambda}\left(\int_{\min(T_{pad}, X_{pad}/c)}^{\infty}  2\lambda we^{-2\lambda w}\mathrm{d}w\right) f(\lambda) \mathrm{d}\lambda \nonumber \\
&= \frac{c\left(\Var(L') + \mathbb{E}\left[L'\right]^{2}\right)}{2} \int_{0}^{\infty}\frac{1}{\lambda^{2}} \left(\int_{2\lambda\min(T_{pad}, X_{pad}/c)}^{\infty} u e^{-u}\mathrm{d}u\right) f(\lambda) \mathrm{d}\lambda \text{ where } u = 2\lambda w, \nonumber \\
&=  \frac{c\left(\Var(L') + \mathbb{E}\left[L'\right]^{2}\right)}{2} \int_{0}^{\infty}\frac{1}{\lambda^{2}} \left(B\lambda +1 \right)e^{-B\lambda} f(\lambda) \mathrm{d}\lambda \text{ where } B = 2\min(T_{pad}, X_{pad}/c), \nonumber \\
&= \frac{c\beta^{\alpha}\left(\Var(L') + \mathbb{E}\left[L'\right]^{2}\right)}{2(\alpha - 1)}\left[\frac{B}{\left(\beta + B\right)^{\alpha -1}} + \frac{1}{(\alpha - 2)\left(\beta + B\right)^{\alpha - 2}} \right]. \label{eqn:MSEc}
\end{align}
From (\ref{eqn:MSEc}), we see that the MSE upper bound converges to zero as $T_{pad}$ and $X_{pad}$ increases. In addition, the rate of convergence increases as $\alpha$ increases. This is in line with the fact that $f(\lambda)$ places more probability weight on larger $\lambda$ values and large $\lambda$ values lead to lower kernel truncation error.
\end{Eg}
\vspace{2mm}

As mentioned in Remark 1 of \cite{BM2017} for the simulation of compound Poisson CARMA random fields, we can approximate the first and second moments of other seed distributions by varying the rate of the compound Poisson L\'evy seed and the jump distribution. For example, if we allow for positive and negative jumps so that the mean jump size is zero and $\mathbb{E}\left[Z_{k}^{2}\right] = \sigma^{2}/\mu$ for $\mu$ large, we obtain an approximation of a Gaussian L\'evy seed with mean zero and variance $\sigma^{2}$. More work needs to be done to see what kind of error is incurred by this additional approximation.
\\
In Figure \ref{fig:psngauapprox}, we set $\mu = 40$ and use a Gaussian jump distribution with mean zero and $\sigma^{2} = \alpha/\beta = 3$ so that we have the same L\'evy seed variance as in the previous case. By virtue of the higher rate parameter, we have more jumps in the same simulation domain (as denoted by the black dots) of smaller size due to the smaller standard deviation. As a result, we get an approximation of the continuous Gaussian L\'evy basis.

\clearpage

\begin{figure}[tbp]
\centering
\caption{Series and autocorrelation function (ACF) plots for: (a)-(b) $Y_{t}(100)$ and (c)-(d) $Y_{100}(x)$. Each time lag is a unit. In the ACF plots, the black curves represent the theoretical ACFs.}
\label{fig:TSACF2}
\includegraphics[width = 4in, height = 4in, trim = 0.2in 0.2in 0.2in 0in]{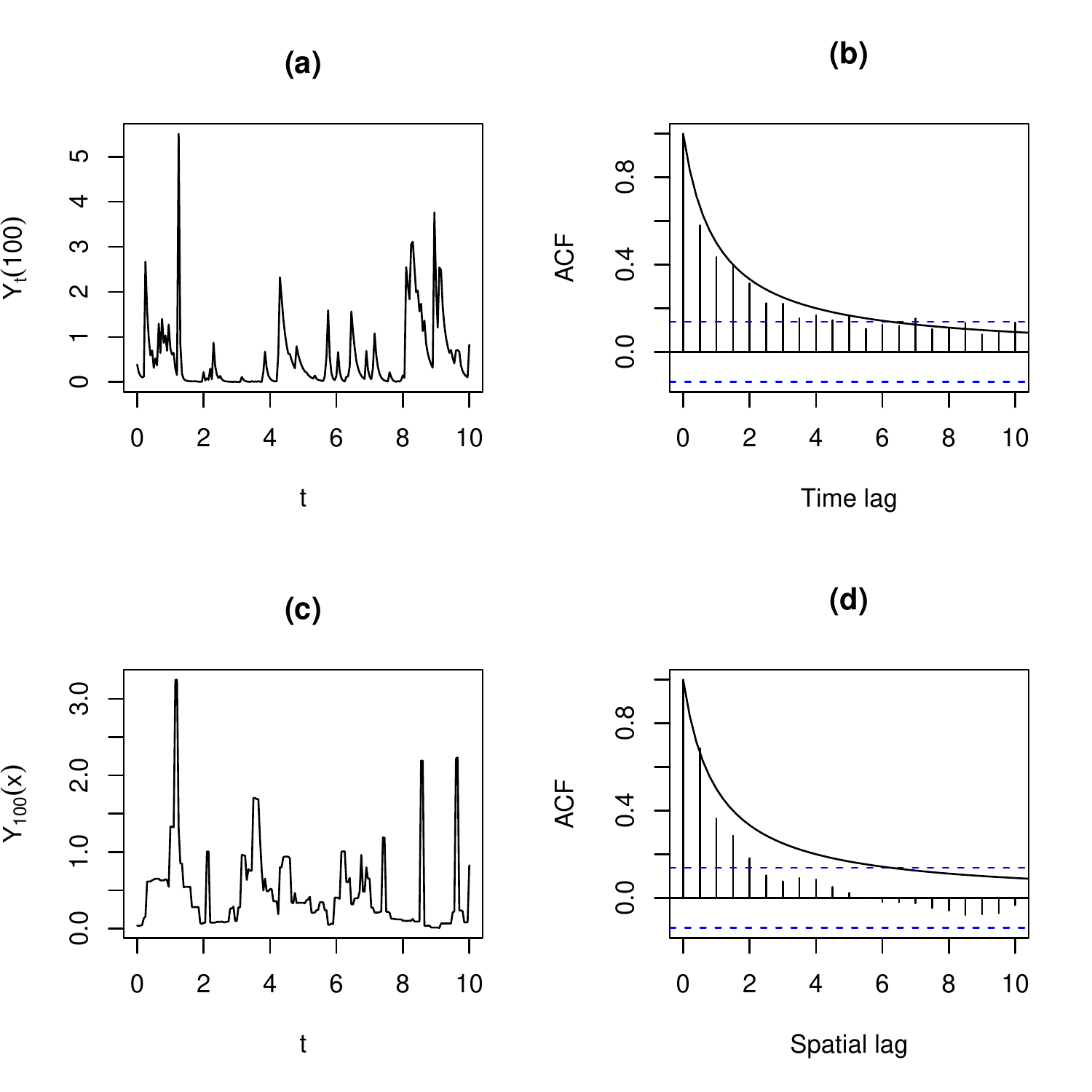}
\end{figure}

\begin{figure}[tbp]
\centering
\caption{Gaussian L\'evy seed approximation: (a) jumps in the extended domain $[-40, 41]\times[-40, 1]$ (in red: jumps in the simulation domain); (b) heat plot over $[0, 1]\times[0, 1]$ for one simulation from a canonical spatio-temporal shot noise process; and (c) the corresponding perspective plot. Here, the grid size is set to $0.01$ units, the rate parameter of the underlying Poisson process is $\mu = 40$ and the jumps are normally distributed jumps with mean zero and standard deviation $\sqrt{\frac{\alpha}{\beta\mu}} = \sqrt{3/10}$; These parameter settings mean that the L\'evy seed variance is equal to that corresponding to the Gamma distributed jumps in Figure \ref{fig:psn}, i.e.~$\alpha/\beta$.}
\label{fig:psngauapprox}
\includegraphics[width = 6in, height = 2.4in, trim = 1in 0.4in 0.4in 0in]{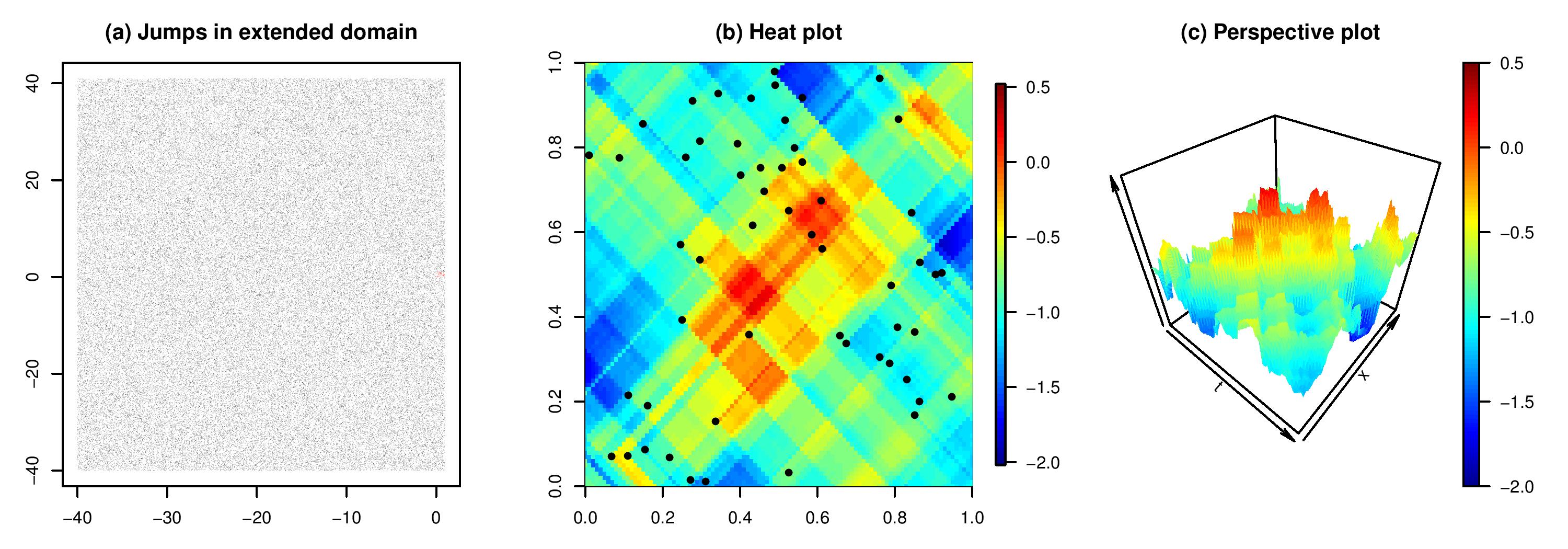}
\end{figure}

\clearpage

\begin{Rem}
If we are using the simulation algorithm as an approximation for MSTOU processes driven by L\'evy bases other than a compound Poisson one, the approach is related to but not exactly the same as the approximation of infinite activity L\'evy processes by a compound Poisson process in \cite{CT2004} since the latter requires a drift term. In the finite variation case (e.g.~the IG and Gamma basis), the L\'evy-It\^o decomposition suggests that we can simulate the process as a sum of a drift term and a compound Poisson process. When the L\'evy density $\nu(z)$ exists, the drift term includes the expectation of jumps less than $\epsilon>0$ which is given by $\int_{\epsilon}^{\infty} z\nu(z) \mathrm{d}z$ while the compound Poisson process has intensity $U(\epsilon) = \int_{\epsilon}^{\infty} v(z) \mathrm{d}z$ and jump size distribution $p^{\epsilon}(z) = \frac{\nu(z)\mathbf{1}_{z\geq \epsilon}}{U(\epsilon)}$. As shown in Proposition 6.1 of \cite{CT2004}, the error incurred by this approximation can be expressed in terms of $\epsilon$. 
\end{Rem}
\vspace{2mm}
\Rem{As seen from Algorithm \ref{alg:psn}, the expected number of iterations required to generate a data set is $\mu\Leb(D)\times n \times m$. This means that the speed of the simulation algorithm depends on the rate parameter of the Poisson process, the extent of the padding, the space-time region we want to cover and the number of simulation grid points. While the first three  parameters determine the number of jumps observed and hence the number of additions required for each process value, the latter determines the number of process values to be generated. On a PC with an Intel$^{\circledR}$ Core\texttrademark i7-3770 CPU Processor @ 3.40GHz, 8GB of RAM and Windows 8.1 64-bit, the data for Figure \ref{fig:psn} took about eight minutes to generate while that for Figure \ref{fig:psngauapprox} took about an hour.}

\section{Two-step iterated GMM estimation} \label{sec:Infer}

In this section, we apply the two-step iterated GMM to MSTOU processes. As mentioned in Section 3 of \cite{STW2015} for supOU processes, this is a semi-parametric estimation method since we conduct inference based on second order moments. However, if we assume a particular distribution for the L\'evy seed $L'$ which is characterised by two parameters, we can estimate these parameters directly.

\subsection{The method}\label{sec:GMM}

For illustrative purposes, we focus on the case in Example \ref{eg:canon1d}. We are interested in estimating $\bm{\beta} = (\alpha, \beta, c, \mathbb{E}\left[L' \right], \Var\left(L' \right)) \in\Theta$. Suppose that we have data on an $N\times N$ space-time grid with origin $(x_{0}, t_{0})$ and grid size $\triangle>0$, we define the vector:

\begin{equation*}
Y_{t}(x)^{(m)} := (Y_{t}(x), Y_{t}(x + \triangle), \dots, Y_{t}(x + m\triangle),  \dots, Y_{t+ \triangle}(x), \dots, Y_{t+ m\triangle}(x)),
\end{equation*}

for $t \in \{t_{0}, \dots, t_{0} + (N-m)\triangle\}$ and $x \in \{x_{0}, \dots, x_{0} + (N-m)\triangle\}$. Next, we construct the following moment function:
\begin{align}
f_{Y}(Y_{t}(x)^{(m)}, \bm{\beta}) &= \begin{pmatrix}
f_{\mathbb{E}}(Y_{t}(x)^{(m)}, \bm{\beta}) \\
f_{\Var}(Y_{t}(x)^{(m)}, \bm{\beta}) \\
f_{X, 1}(Y_{t}(x)^{(m)}, \bm{\beta}) \\
\vdots \\
f_{X, m}(Y_{t}(x)^{(m)}, \bm{\beta}) \\
f_{T, 1}(Y_{t}(x)^{(m)}, \bm{\beta}) \\
\vdots \\
f_{T, m}(Y_{t}(x)^{(m)}, \bm{\beta})
\end{pmatrix}, \label{eqn:mcon}\\
\text{where } f_{\mathbb{E}}(Y_{t}(x)^{(m)}, \bm{\beta}) &= Y_{t}(x) - \frac{2c \beta^{2}\mathbb{E}\left[L' \right]}{(\alpha - 2)(\alpha - 1)} \nonumber\\
f_{\Var}(Y_{t}(x)^{(m)}, \bm{\beta}) &= Y_{t}(x)^2 - \frac{c\beta^{2} \Var\left(L' \right)}{2(\alpha - 2)(\alpha - 1)} - \left(\frac{2c \beta^{2}\mathbb{E}\left[L' \right]}{(\alpha - 2)(\alpha - 1)} \right)^{2}  \nonumber \\
f_{X, h}(Y_{t}(x)^{(m)}, \bm{\beta}) &= Y_{t}(x)Y_{t}(x + h\triangle) -  \frac{c\beta^{\alpha} \Var\left(L' \right)}{2(\beta+ h\triangle/c)^{\alpha - 2}(\alpha - 2)(\alpha - 1)}  - \left(\frac{2c \beta^{2}\mathbb{E}\left[L' \right]}{(\alpha - 2)(\alpha - 1)} \right)^{2}  \nonumber\\
f_{T, h}(Y_{t}(x)^{(m)}, \bm{\beta}) &= Y_{t}(x)Y_{t+ h\triangle}(x) - \frac{c\beta^{\alpha} \Var\left(L' \right)}{2(\beta+h\triangle)^{\alpha - 2}(\alpha - 2)(\alpha - 1)}  - \left(\frac{2c \beta^{2}\mathbb{E}\left[L' \right]}{(\alpha - 2)(\alpha - 1)} \right)^{2},  \nonumber
\end{align}
where $h = 1, \dots, m$ and $m \geq 2$ is an integer. 
\\
The GMM estimator of $\bm{\beta}$ is given by:
\begin{align*}
\hat{\bm{\beta}}_{N} &= \argmin_{\bm{\beta}} \left\{g_{N, m}(Y, \bm{\beta})\right\}'W_{N}\left\{g_{N, m}(Y, \bm{\beta})\right\},  \\
\text{where } g_{N, m}(Y, \bm{\beta}) &= \frac{1}{(N-m)^2}\sum_{i = 1}^{N-m}\sum_{j = 1}^{N-m}f_{Y}(Y_{t_{0} + i\triangle}(x_{0} +j\triangle)^{(m)}, \bm{\beta}). 
\end{align*}
In the first step of the GMM procedure, we set $W_{N} = I$, the $2(1+m)\times 2(1+m)$ identity matrix to find $\hat{\bm{\beta}}_{1, N}$, the first step estimator. In the second step of the GMM procedure, we set $W_{N}$ to be $\widehat{S}_{N}^{-1}$ where:
\begin{equation*}
\widehat{S}_{N} := \frac{1}{(N-m)^{2}}\sum_{i = 1}^{N-m}\sum_{j = 1}^{N-m}f_{Y}(Y_{t_{0} + i\triangle}(x_{0} +j\triangle)^{(m)}, \hat{\bm{\beta}}_{1, N})f_{Y}(Y_{t_{0} + i\triangle}(x_{0} +j\triangle)^{(m)}, \hat{\bm{\beta}}_{1, N})'.
\end{equation*}
We note that $\widehat{S}_{N}$ is an estimator for $\bar{V}_{N} = N\Var\left( g_{N, m}(Y, \bm{\beta}_{0})\right)$. Improvements can be made by considering the autocorrelation effects. 
\vspace{2mm}
\begin{thm}\label{thm:identifiability}
Let $Y_{t}(x)$ be the MSTOU process defined in Example \ref{eg:canon1d}, $m\geq 2$ be a fixed integer and $f_{Y}(Y_{t}(x)^{(m)}, \bm{\beta})$ be as defined in (\ref{eqn:mcon}). Then, the true parameter vector $\bm{\beta}_{0}$ is identifiable, i.e.~$\mathbb{E}\left[f_{Y}(Y_{t}(x)^{(m)}, \bm{\beta})\right] = \mathbf{0}$ for all $(x, t)$ if and only if $\bm{\beta} = \bm{\beta}_{0}$. 
\end{thm}
\vspace{2mm}
\begin{Rem}
To extend the GMM approach to $d>1$, we can add additional observations corresponding to other spatial directions to $Y_{t}(\mathbf{x})^{(m)}$ and adapt $f_{X, h}$ and $f_{T, h}$ accordingly.
\end{Rem}

\subsection{Simulation experiments}

GMM estimators are known to be consistent and asymptotically normal under certain assumptions \cite[]{Matyas1999}. For example, one typically requires that $W_{N}$ converges to a positive definite matrix and that a central limit theorem (CLT) holds for $f_{Y}(Y_{t}(x)^{(m)}, \bm{\beta})$. These conditions are hard to check in practice. In addition, little work has been done in establishing CLTs for general supOU processes, much less MSTOU processes. So far, it has been shown that CLTs hold for supOU processes when $f(\lambda)$ is a discrete probability distribution with finite support but may not hold under infinite support \cite[]{GLST2016}. It is likely that similar results hold for MSTOU processes but proving them is out of the scope of this paper. Instead, to illustrate our method and strengthen our conjectures about the asymptotic properties, we conduct simulation studies.
\\
We set the simulation domain to $D = [0, 100]\times[0, 100]$, the intensity of the underlying Poisson process to $\mu = 0.2$, the rate parameter of the Gamma distribution for $\lambda$ to $\beta = 1$, the shape parameter of the ambit set to $c = 1$ and use a $N(0, 15)$ jump distribution. The padding extents and grid size are chosen to be $X_{pad} = T_{pad} = 40$ and $\triangle = 0.5$ respectively. $100$ data sets are generated for the short-range dependence case with $\alpha = 5$ and the long-range dependence case with $\alpha = 3$. 
\\
To have a properly overidentified system and avoid high dimensional matrices, we choose $m = 3$ and conduct the two-step GMM estimation as laid out in Section \ref{sec:GMM}. The ``DEoptim'' function of the ``DEoptim'' R package was used to perform global optimisation over the parameter space $[2, 35]\times[0, 35]\times[0, 5]\times[-2.5, 2.5]\times[0, 15]$. Figure \ref{fig:GMMestout} shows box plots of the estimates for the short-range and long-range dependence scenarios in the top and bottom rows respectively. For a closer look at where the majority lie, we have omitted one, five and four outliers for $\hat{\alpha}$, $\hat{\beta}$ and $\hat{c}$ in long-range dependence setting. From the plots, we see that the true parameter values (denoted by the red horizontal lines) lie well within the range of the estimates. We also notice that when the data has short-range dependence, $\alpha$ is never estimated to be lower than $3$, the boundary value for long-range dependence (denoted by the dotted blue line). This ability to distinguish between the two forms of dependence is desirable in practice. However, we also note that it is one-sided since $\hat{\alpha}>3$ for many long-range dependent data sets. There is also some skewness and bias in the estimates which one might expect since we 

\clearpage

\begin{figure}[tbp]
\centering
\caption{Full data without extreme outliers: Box plots of GMM parameter estimates from $100$ simulated data sets. The top row corresponds to the case of short-range dependence ($\alpha = 5$) while the bottom row corresponds to long-range dependence ($\alpha = 3$). The red horizontal lines denote the true parameter vaues and the blue dotted line in Plot (a) denotes $\alpha = 3$, the boundary value for long-range dependence.}
\label{fig:GMMestout}
\includegraphics[width = 5.2in, height = 3.4in, trim = 0.2in 0.2in 0.2in 0in]{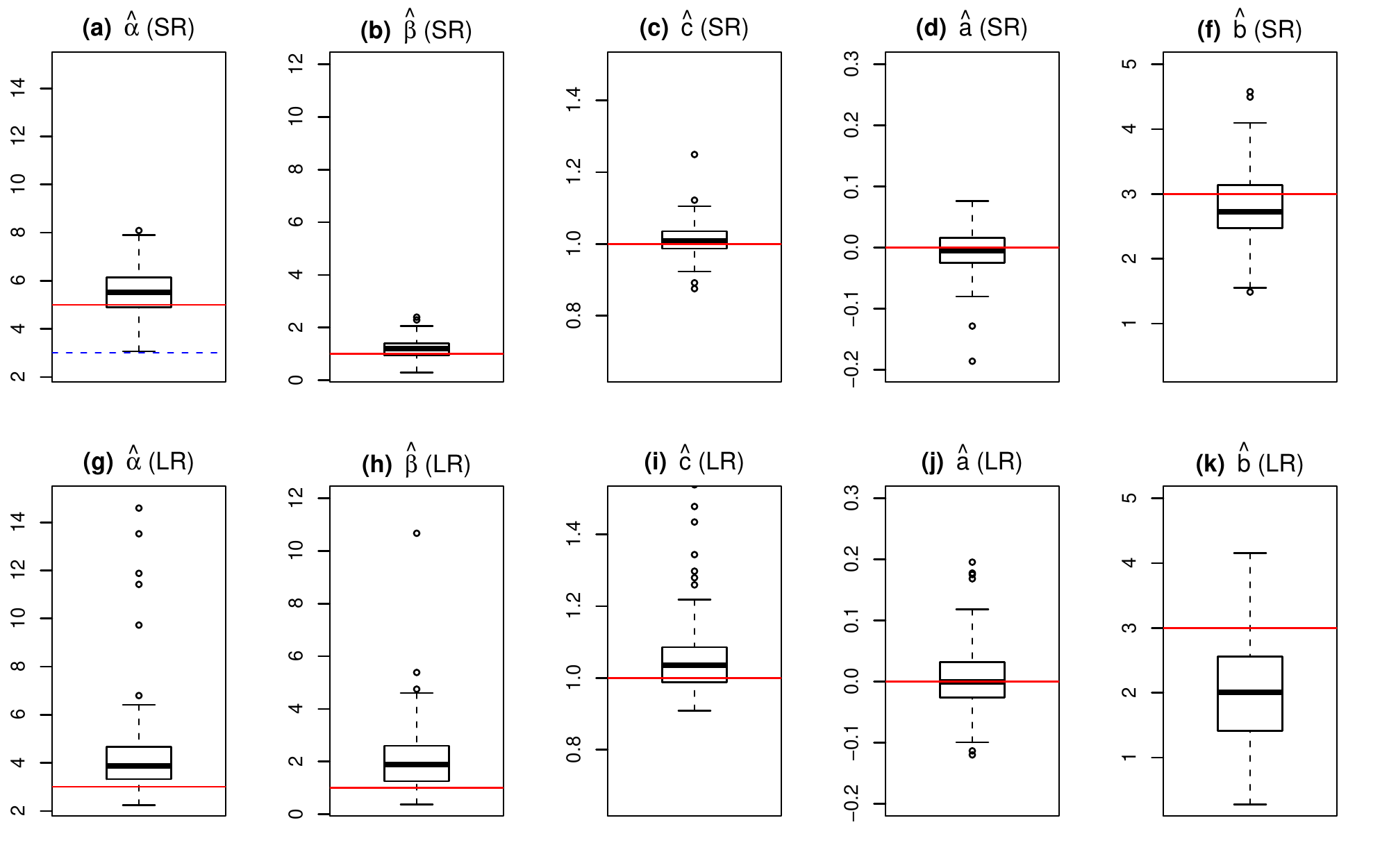}
\end{figure}

\begin{figure}[tbp]
\centering
\caption{Reduced data without extreme outliers: Box plots of GMM parameter estimates from $100$ simulated data sets. The top row corresponds to the case of short-range dependence ($\alpha = 5$) while the bottom row corresponds to long-range dependence ($\alpha = 3$). The red horizontal lines denote the true parameter vaues and the blue dotted line in Plot (a) denotes $\alpha = 3$, the boundary value for long-range dependence.}
\label{fig:GMMestredout}
\includegraphics[width = 5.2in, height = 3.4in, trim = 0.2in 0.2in 0.2in 0in]{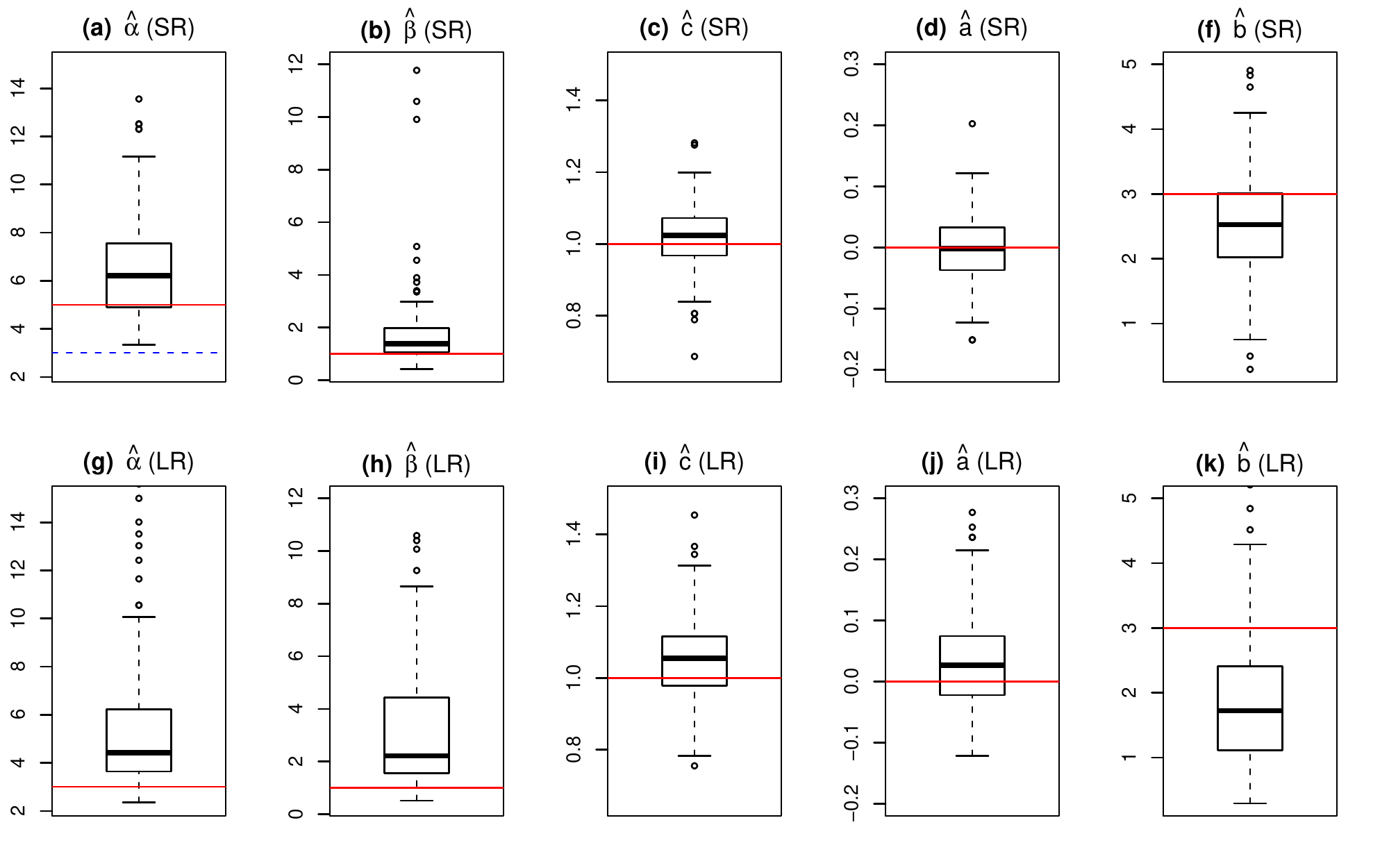}
\end{figure}

\clearpage

\begin{figure}[tbp]
\centering
\caption{Full data: Normal QQ plots of GMM parameter estimates from $100$ simulated data sets. The top row corresponds to the case of short-range dependence ($\alpha = 5$) while the bottom row corresponds to long-range dependence ($\alpha = 3$).}
\label{fig:GMMqq}
\includegraphics[width = 6.6in, height = 3.4in, trim = 0.2in 0.2in 0.2in 0in]{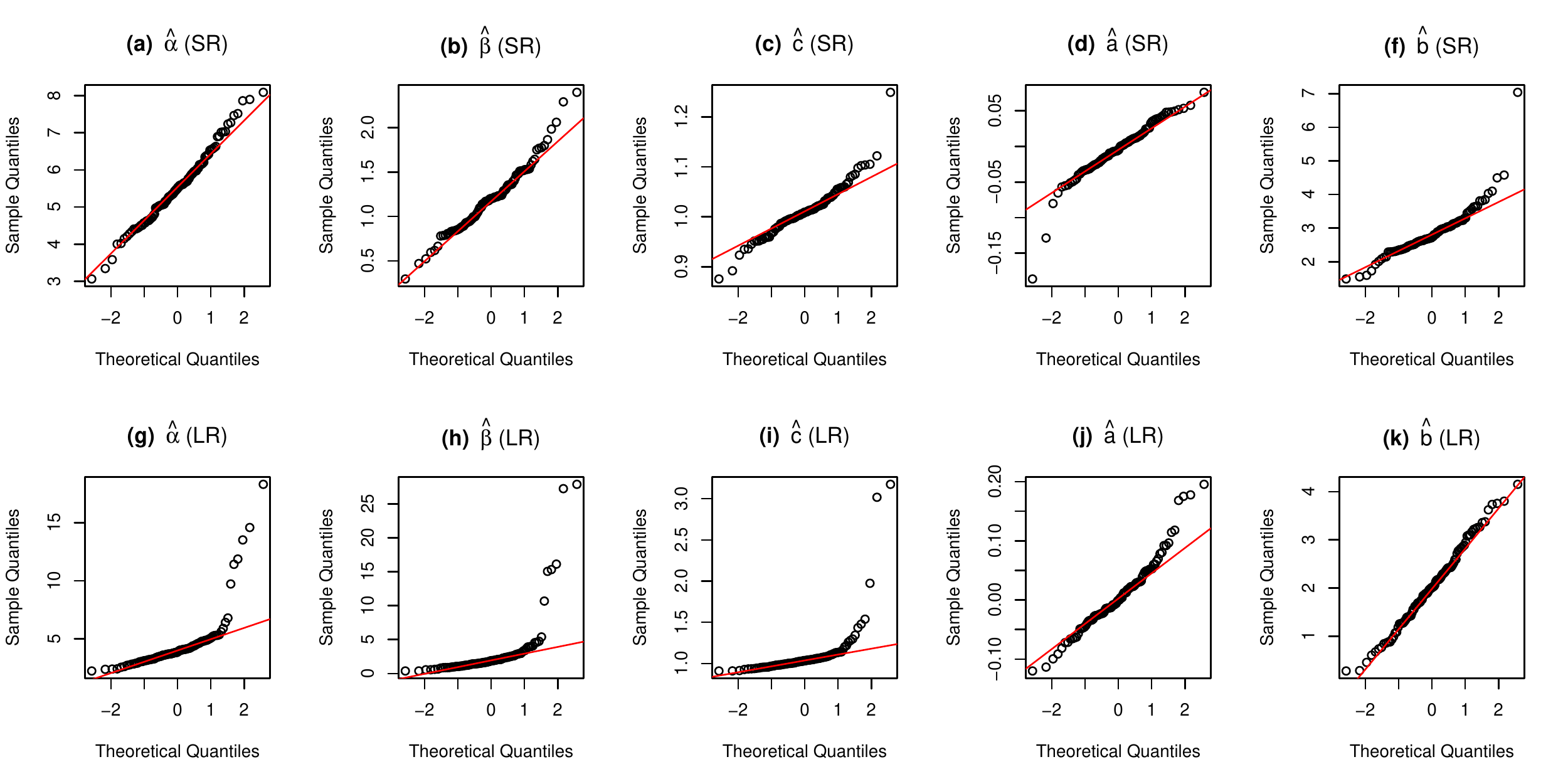}
\end{figure}

are not in the asymptotic regime.
\\
To comment on the possibility of the consistency of our estimators, we repeat our estimation on subsets of our simulated data over the reduced space-time region $[25, 75] \times [50, 100]$. The box plots of the results are shown in Figure \ref{fig:GMMestredout}. As before, a few outliers have been removed to enable us to zoom into majority of the estimates. From the plots, we see that the ranges and bias of the estimates are larger than those for the full data sets. Since the estimates become closer to the true values as more data is included in inference, it seems that consistency does hold for both dependence scenarios. It is also interesting to note that even for the reduced data sets, $\hat{\alpha}$ does not drop below $\alpha = 3$ under short-range dependence.
\\
Next, we look at the normal quantile-quantile (QQ) plots for the full data sets in Figure \ref{fig:GMMqq}. Apart from $\hat{b}$, there are stronger deviations from normality under the long-range dependence than short-range dependence. Just as how the asymptotic distributions for partial sums of one-dimensional transformations of Gaussian processes with finite second moments depend on the Hurst parameter (see Theorem 3.1 of \cite{Beran1994}), it is reasonable to hypothesize that the asymptotic distribution of the sample mean of a MSTOU process will depend on the strength of the dependence.

\section{Conclusion and further work} \label{sec:Conclusion}

The mixed spatio-temporal Ornstein-Uhlenbeck (MSTOU) process is an extension of the STOU process studied in \cite{BS2003} and \cite{NV2016}. While the highlight of this set up is the ability to encompass both short-range and long-range dependence, the MSTOU process also retains the ability to create non-separable spatio-temporal covariances and flexible spatial covariances. This was illustrated for an isotropic class of MSTOU processes, known as the $g$-class, in Section \ref{sec:gclass}. 
\\
After developing the theory for MSTOU processes in Sections \ref{sec:Prelim} and \ref{sec:Prop}, we presented a simulation algorithm in Section \ref{sec:Sim}. Unlike the discrete convolution algorithms for STOU process in \cite{NV2016}, our algorithm does not suffer from the kernel discretisation errors. Instead, the simulation error depends on the kernel truncation: for $g$-class processes with compound Poisson L\'evy bases, an upper bound for the mean squared error was shown to shrink to zero as the simulation padding extents increase to infinity. As already mentioned, it will be useful to better understand the implications of approximating other L\'evy bases using our simulation algorithm.
\\
Since we derived the stationarity and second order moments of our processes in Section \ref{sec:Prop}, we applied the two-step iterated generalised method of moments (GMM) to an MSTOU process in Section \ref{sec:Infer}. Promising results were obtained from the simulation experiments. These support the view that while consistency of the estimators may hold, asymptotic normality may or may not hold depending on the strength of the dependence. More work needs to be done in order to formally establish these asymptotic properties. In particular, it would be useful to determine the asymptotic distributions of the sample averages of MSTOU processes. 
\\
So far, we have focused mostly on isotropic MSTOU processes. Extending our results to anisotropy via geometric or coordinate-wise means is an interesting direction for further research. While the former assumes isotropy for transformed space-time coordinates, the latter assumes isotropy in individual spatial directions only. 

\section*{Appendix:} 

\begin{proof}[Proof of Corollary \ref{cor:intcon}]
This is an extension of the proof of existence for canonical STOU processes on pages 3-4 of the supplementary material of \cite{NV2016}. 
\end{proof}

\begin{proof}[Proof of Theorem \ref{thm:GCF}]
This follows the proofs for Proposition 1 and 5 in \cite{BBV2012} with $h_{A}$ being defined differently to account for space-time, the $\lambda$ parameter space and the definition of an MSTOU process. Based on our assumptions and Fubini's theorem:
\begin{equation*}
v(Y) =  \int_{\mathbb{R}^{d}\times\mathbb{R}} Y_{t}(\mathbf{x}) v(\mathrm{d}\mathbf{x}, \mathrm{d}t) = \int_{S} \int_{\mathbb{R}^{d}\times\mathbb{R}} \mathbf{1}_{A}(\bm{\xi} - \mathbf{x}, s-t) \exp(-\lambda(t-s)) v(\mathrm{d}\mathbf{x}, \mathrm{d}t) L(\mathrm{d}\bm{\xi}, \mathrm{d}s, \mathrm{d}\lambda).
\end{equation*}
Using Proposition 2.6 of \cite{RR1989}, we obtain the expression for the CGF of $v(Y)$. 
\end{proof}

\begin{proof}[Proof of Theorem \ref{thm:tsstation}]
This is analogous to the proof of Theorem 3 in \cite{NV2016} with $ h_{A}(\bm{\xi}, s, \lambda)$ replacing $ h_{A}(\bm{\xi}, s)$.
\end{proof}

\begin{proof}[Proof of Corollary \ref{Cor:meancov}]
For information on the bivariate distributions, we use the result in Theorem \ref{thm:GCF} with $v(\mathrm{d}\mathbf{x}, \mathrm{d}t) = \theta_{1}\delta_{t_{1}}(\mathrm{d}t)\delta_{\mathbf{x}_{1}}(\mathrm{d}\mathbf{x}) +  \theta_{2}\delta_{t_{2}}(\mathrm{d}t)\delta_{\mathbf{x}_{2}}(\mathrm{d}\mathbf{x})$ where $(\mathbf{x}_{1}, t_{1})$ and $(\mathbf{x}_{2}, t_{2})$ denote arbitary locations in space-time. We also set $\theta = 1$. With these specifications, we find that:
\begin{equation*}
h_{A}(\bm{\xi} ,s, \lambda) = \int_{\mathbb{R}^{d}\times\mathbb{R}} \mathbf{1}_{A}(\bm{\xi} - \mathbf{x}, s-t) \exp(-\lambda(t-s)) v(\mathrm{d}\mathbf{x}, \mathrm{d}t) = \sum_{i=1}^{2} \theta_{i}\mathbf{1}_{A}(\bm{\xi} - \mathbf{x}_{i}, s-t_{i}) \exp(-\lambda(t_{i}-s)).
\end{equation*}
Since we can obtain the covariance structure by differentiating the bivariate CGF with respect to $\theta_{1}$ and $\theta_{2}$, and setting $\theta_{1} = \theta_{2} = 0$, we are interested in the cross terms. The first term in (\ref{eqn:MSTOULevy}) does not contain any cross terms:
\begin{equation*}
i\theta a\int_{S}  h_{A}(\bm{\xi},s, \lambda)f(\lambda)\mathrm{d}\bm{\xi} \mathrm{d}s\mathrm{d}\lambda =  i a\sum_{i=1}^{2}\theta_{i}\int_{0}^{\infty} \int_{A_{t_{i}}(\mathbf{x}_{i})} \exp(-\lambda(t_{i}-s))\mathrm{d}\bm{\xi}\mathrm{d}s f(\lambda) \mathrm{d}\lambda. 
\end{equation*}
A cross term appears in the second term of (\ref{eqn:MSTOULevy}):
\begin{align*}
-\frac{1}{2}\theta^{2}b\int_{S}  h^{2}_{A}(\bm{\xi},s, \lambda)f(\lambda)\mathrm{d}\bm{\xi}\mathrm{d}s\mathrm{d}\lambda &=  -\frac{1}{2}b\left[\sum_{i=1}^{2} \theta_{i}^{2} \int_{0}^{\infty} \int_{A_{t_{i}}(\mathbf{x}_{i})} \exp(-2\lambda(t_{i}-s))\mathrm{d}\bm{\xi}\mathrm{d}s f(\lambda) \mathrm{d}\lambda \right. \\
&\left.+  2\theta_{1}\theta_{2}\int_{0}^{\infty} \int_{A_{t_{1}}(\mathbf{x}_{1})\cap A_{t_{2}}(\mathbf{x}_{2})} \exp(-\lambda(t_{1}+t_{2}-2s))\mathrm{d}\bm{\xi}\mathrm{d}s f(\lambda) \mathrm{d}\lambda \right],
\end{align*}
Differentiating the cross term respect to $\theta_{1}$ and $\theta_{2}$, and setting $\theta_{1} = \theta_{2} = 0$, we have:
\begin{equation}
-b\int_{0}^{\infty} \int_{A_{t_{1}}(\mathbf{x}_{1})\cap A_{t_{2}}(\mathbf{x}_{2})} \exp(-\lambda(t_{1}+t_{2}-2s))\mathrm{d}\bm{\xi}\mathrm{d}s f(\lambda) \mathrm{d}\lambda. \label{eqn:con2t} 
\end{equation}
By splitting the integration regions into $A_{t_{1}}(\mathbf{x}_{1})\backslash A_{t_{2}}(\mathbf{x}_{2})$,  $A_{t_{2}}(\mathbf{x}_{2})\backslash A_{t_{1}}(\mathbf{x}_{1})$ and $A_{t_{1}}(\mathbf{x}_{1})\cap A_{t_{2}}(\mathbf{x}_{2})$, we can express the last term in (\ref{eqn:MSTOULevy}) as: 
\begin{align}
&\int_{S} \int_{\mathbb{R}} \left(\exp(i\theta h_{A}(\bm{\xi},s, \lambda)z) - 1 - i\theta h_{A}(\bm{\xi},s, \lambda)z\mathbf{1}_{|z|\leq 1}\right) \nu(\mathrm{d}z)  f(\lambda)\mathrm{d}\bm{\xi}\mathrm{d}s\mathrm{d}\lambda \nonumber \\ 
&= \int_{0}^{\infty}\int_{\mathbb{R}}\int_{A_{t_{1}}(\mathbf{x}_{1})\backslash A_{t_{2}}(\mathbf{x}_{2})}  \left(\exp(i\theta_{1}\exp(-\lambda(t_{1}-s))z) - 1 - i\theta_{1}\exp(-\lambda(t_{1}-s)) z\mathbf{1}_{|z|\leq 1}\right)  \mathrm{d}\bm{\xi}\mathrm{d}s\nu(\mathrm{d}z) f(\lambda)\mathrm{d}\lambda \nonumber \\
&+  \int_{0}^{\infty}\int_{\mathbb{R}}\int_{A_{t_{2}}(\mathbf{x}_{2})\backslash A_{t_{2}}(\mathbf{x}_{1})}  \left(\exp(i\theta_{2}\exp(-\lambda(t_{2}-s))z) - 1 - i\theta_{2}\exp(-\lambda(t_{2}-s)) z\mathbf{1}_{|z|\leq 1}\right)  \mathrm{d}\bm{\xi}\mathrm{d}s\nu(\mathrm{d}z) f(\lambda)\mathrm{d}\lambda \nonumber \\
&+  \int_{0}^{\infty}\int_{\mathbb{R}}\int_{A_{t_{1}}(\mathbf{x}_{1})\cap A_{t_{2}}(\mathbf{x}_{2})}  \left(\exp\left(i\left[\theta_{1}\exp(-\lambda(t_{1}-s)) + \theta_{2}\exp(-\lambda(t_{2}-s))\right]z\right) - 1 \right. \nonumber \\
&\left. - i\left[\theta_{1}\exp(-\lambda(t_{1}-s)) + \theta_{2}\exp(-\lambda(t_{2}-s))\right] z\mathbf{1}_{|z|\leq 1}\right)  \mathrm{d}\bm{\xi}\mathrm{d}s\nu(\mathrm{d}z) f(\lambda)\mathrm{d}\lambda. \label{eqn:lastterm}
\end{align}
When we differentiate with respect to $\theta_{1}$ and $\theta_{2}$, set $\theta_{1} = \theta_{2} = 0$, the first two terms of (\ref{eqn:lastterm}) equal to zero. The same procedure on the last term gives:
\begin{align*}
- \int_{0}^{\infty}\int_{\mathbb{R}}\int_{A_{t_{1}}(\mathbf{x}_{1})\cap A_{t_{2}}(\mathbf{x}_{2})} z^{2}\exp(-\lambda(t_{1} + t_{2} - 2s)) \mathrm{d}\bm{\xi}\mathrm{d}s\nu(\mathrm{d}z) f(\lambda)\mathrm{d}\lambda. 
\end{align*}
The required expression of the covariance function is obtained by adding this to (\ref{eqn:con2t}) and multiplying the result by $-1$. To obtain the expression for the mean of $Y$, we differentiate each of the term in (\ref{eqn:MSTOULevy}) by either $\theta_{1}$ or $\theta_{2}$ and set $\theta_{1} = \theta_{2} = 0$. The mean is then given by multiplying the result by $-i$. 
\end{proof}

\begin{proof}[Proof of Corollary \ref{cor:gcon}]
The conditions follow from (\ref{eqn:simicon}) since the temporal cross-section of $A_{t}(\mathbf{x})$ corresponds to the $d$-dimensional sphere with centre $(\mathbf{x}, s)$ for $s\leq t$ and radius $g(|t-s|)$. 
\end{proof}

\begin{proof}[Proof of Theorem \ref{thm:giso}]
Fix $t\in\mathbb{R}$. From (\ref{eqn:Ycov}), the spatial covariance of $Y$ is given by:
\begin{align*}
\Cov(Y_{t}(\mathbf{x}), Y_{t}(\mathbf{x}+d_{\mathbf{x}})) &= \Var(L') \int_{0}^{\infty} \int_{A_{t}(\mathbf{x})\cap A_{t}(\mathbf{x} + d_{\mathbf{x}})} \exp(-2\lambda(t-s))\mathrm{d}\bm{\xi}\mathrm{d}s f(\lambda) \mathrm{d}\lambda,
\end{align*}
where $d_{\mathbf{x}}\in\mathbb{R}^{d}$ denotes the spatial displacement vector while $L'$ denotes the L\'evy seed of $Y$. \\
Suppose first that $d = 1$, i.e.~we have one dimensional space. Without loss of generality, let $d_{x} \geq 0$, then:
\begin{align}
\Cov(Y_{t}(x), Y_{t}(x+d_{x})) &= \Var(L') \int_{0}^{\infty} \int_{\infty}^{t - g^{-1}(|d_{x}|/2)}\int^{x+g(|t-s|)}_{x+ d_{x} - g(|t-s|)} \exp(-2\lambda(t-s))\mathrm{d}\bm{\xi}\mathrm{d}s f(\lambda) \mathrm{d}\lambda \nonumber \\
&= \Var(L') \int_{0}^{\infty} \int_{\infty}^{t - g^{-1}(|d_{x}|/2)}(2g(|t-s|)- |d_{x}|) \exp(-2\lambda(t-s))\mathrm{d}s f(\lambda) \mathrm{d}\lambda \label{eqn:scovg} \\
&=  \Var(L') \int_{0}^{\infty} \int_{g^{-1}(|d_{x}|/2)}^{\infty}(2g(w)- |d_{x}|) \exp(-2\lambda w)\mathrm{d}w f(\lambda) \mathrm{d}\lambda, \nonumber
\end{align}
where $w = t-s$. Note that $t - g^{-1}(|d_{x}|/2)$ denotes the largest temporal coordinate of $A_{t}(\mathbf{x})\cap A_{t}(\mathbf{x} + d_{\mathbf{x}})$ since $A_{t}(\mathbf{x})$ is radially symmetric and translation invariant. Since the spatial covariance of $Y$ is a function of the spatial distance $|d_{x}|$, $Y$ is isotropic in space.
\\
For general $d\in\mathbb{N}$, replace $(2g(|t-s|)- d_{x})$ in (\ref{eqn:scovg}) with the volume of the intersection of two $d$-spheres with the same radius $g(|t-s|)$ and centres at $\mathbf{x}$ and $\mathbf{x} + d_{\mathbf{x}}\in\mathbb{R}^{d}$. This can be written as the volume of two identical spherical caps \cite[]{Li2011}:
\begin{equation*}
\frac{\pi^{(d-1)/2}}{\Gamma\left(\frac{d-1}{2} + 1\right)} g^{d}(|t-s|)B\left(1 - \left(\frac{|d_{\mathbf{x}}|}{2g(|t-s|)}\right)^{2}; \frac{d+1}{2}, \frac{1}{2}\right),
\end{equation*} 
where $B$ denotes the incomplete beta function. Since this quantity is a function of $|d_{\mathbf{x}}|$, $Y$ is isotropic in space for general $d\in\mathbb{N}$.
\end{proof}

\begin{proof}[Proof of Theorem \ref{thm:CARp}]
From Remark 2 of \cite{BDY2012}, the eigenvectors of $A$ are $\mathbf{v}_{i} = (1, \eta_{i}, \eta_{i}^{2} \dots, \eta_{i}^{p-1})^{T}$ for $i = 1, \dots, p$. With $V = (v_{1} \dots v_{p})$, we can write:
\begin{equation}
\exp\big(A(t-s)\big)\mathbf{e}_{p} = V\begin{pmatrix} \exp(\eta_{1}(t-s)) & 0 & \mathbf{0} \\ \vdots & \ddots & \vdots \\ \mathbf{0} & 0 & \exp(\eta_{p}(t-s))\end{pmatrix}V^{-1}\mathbf{e}_{p}. \label{eqn:interCAR}
\end{equation}
Since $V$ is the transpose of a Vandermonde matrix, the term ``$V^{-1}\mathbf{e}_{p}$'' which corresponds to the last column of $V^{-1}$ now corresponds to the last row of the Vandermonde matrix inverse. We obtain the required result by using the formulae for these matrix entries in Exercise 40 in Section 1.2.3 of \cite{Knuth1997} and substituting the corresponding expression for (\ref{eqn:interCAR}) in the definition of $Y_{t}(\mathbf{x})$.  
\end{proof}

\begin{proof}[Proof of Theorem \ref{thm:MSE}]
By bounding the MSE by that for the boundaries of our simulation domain, we have:
\begin{align*}
\mathbb{E}\left[\left(Y_{t}(\mathbf{x}) - Z_{t}(\mathbf{x})\right)^{2}\right] &\leq \mathbb{E}\left[\left(\int_{0}^{\infty}\int_{A_{t}(\mathbf{x})\backslash [\mathbf{x} - X_{pad}, \mathbf{x} + X_{pad}]\times[t-T_{pad}, t]} e^{-\lambda(t-s)} L(\mathrm{d}\bm{\xi}, \mathrm{d}s, \mathrm{d}\lambda)\right)^{2}\right] \\
&= \left(\Var(L') + \mathbb{E}\left[L'\right]^{2}\right) \int_{0}^{\infty}\int_{A_{t}(x)\backslash  [\mathbf{x} - X_{pad}, \mathbf{x} + X_{pad}]\times[t-T_{pad}, t]} e^{-2\lambda(t-s)} f(\lambda) \mathrm{d}\bm{\xi}\mathrm{d}s\mathrm{d}\lambda \\
&\leq \frac{\pi^{d/2}\left(\Var(L') + \mathbb{E}\left[L'\right]^{2}\right)}{\Gamma\left(\frac{d}{2} + 1\right)} \int_{0}^{\infty}\left(\int_{\min(T_{pad}, g^{-1}(X_{pad}))}^{\infty}  g^{d}(w)e^{-2\lambda w}\mathrm{d}w\right) f(\lambda) \mathrm{d}\lambda, 
\end{align*}
where $w = t-s$, $ [\mathbf{x} - X_{pad}, \mathbf{x} + X_{pad}] =  [x_{1} - X_{pad}, x_{1} + X_{pad}]\times \dots \times  [x_{d} - X_{pad}, x_{d} + X_{pad}]$ and we have used the fact that the temporal cross-section of the ambit set is the $d$-dimensional sphere centred at $\mathbf{x}$ with radius $g(|t-s|)$.  
\end{proof}

\begin{proof}[Proof of Theorem \ref{thm:identifiability}]
This is similar to the arguments used to establish identifiability of the GMM estimator for the supOU process in Proposition 3.3 of \cite{STW2015}. When $m\geq 2$, we can use the temporal correlations at two different time lags to identify $\alpha$ and $\beta$ uniquely. These can then be used to determine $c$ from the spatial correlation. Lastly, $\mathbb{E}\left[L'\right]$ and $\Var\left(L'\right)$ can be found through the mean and variance.
\end{proof}

\section*{Acknowledgements}

M. Nguyen is grateful to Imperial College for her PhD scholarship which supported this research. A.E.D. Veraart acknowledges financial support by a Marie Curie FP7 Integration Grant (grant agreement number PCIG11-GA-2012-321707) within the 7th European Union Framework Programme.

\bibliographystyle{agsm} 
\bibliography{refsMSTOU}

\vspace{2mm}
Michele Nguyen, Department of Mathematics, Imperial College London, 180 Queen's Gate, SW7 2AZ London, UK. \\
Email: michele.nguyen09@imperial.ac.uk

\end{document}